\documentclass[12pt]{article}
\usepackage[colorlinks,citecolor=blue,urlcolor=blue,breaklinks]{hyperref}
\usepackage{xcolor, latexsym, amsmath,amssymb,amsfonts,amsthm,bbm,bm,enumitem,xcolor,soul,graphicx, dsfont,enumitem,natbib,threeparttable,caption,subcaption, comment}
\usepackage[normalem]{ulem}
\usepackage[ruled, vlined]{algorithm2e}
\usepackage{multirow}
\SetKwInput{KwSet}{Set}

\newtheorem{thm}{Theorem}
\newtheorem{prop}[thm]{Proposition}
\newtheorem{lemma}[thm]{Lemma}
\newtheorem{cor}[thm]{Corollary}

\newtheorem*{remark}{Remark}

\DeclareMathOperator*{\argmax}{argmax}
\DeclareMathOperator*{\sargmin}{sargmin}
\DeclareMathOperator*{\sargmax}{sargmax}
\DeclareMathOperator*{\largmax}{largmax}
\DeclareMathOperator*{\largmin}{largmin}

\newcommand{\bdot}{\bm{\cdot}}

\title{Inference in high-dimensional online changepoint detection}
\author{Yudong Chen$^{*,\dagger}$, Tengyao Wang$^\dagger$ and Richard J. Samworth$^*$ \\
$^*$ Statistical Laboratory, University of Cambridge \\
$^\dagger$ Department of Statistics, London School of Economics
}
\date{(February 17, 2023)}

\begin{document}
	\maketitle
	
	\begin{abstract}
	We introduce and study two new inferential challenges associated with the sequential detection of change in a high-dimensional mean vector.  First, we seek a confidence interval for the changepoint, and second, we estimate the set of indices of coordinates in which the mean changes.  We propose an online algorithm that produces an interval with guaranteed nominal coverage, and whose length is, with high probability, of the same order as the average detection delay, up to a logarithmic factor.  The corresponding support estimate enjoys control of both false negatives and false positives.  Simulations confirm the effectiveness of our methodology, and we also illustrate its applicability on the US excess deaths data from 2017--2020.
	\end{abstract}
	
	\section{Introduction}
	\label{Sec:Intro}
	
	The real-time monitoring of evolving processes has become a characteristic feature of 21st century life.  Watches and defibrillators track health data, Covid-19 case numbers are reported on a daily basis and financial decisions are made continuously based on the latest market movements.  Given that changes in the dynamics of such processes are frequently of great interest, it is unsurprising that the area of changepoint detection has undergone a renaissance over the last 5--10 years.  
	
	One of the features of modern datasets that has driven much of the recent research in changepoint analysis is high dimensionality, where we monitor many processes simultaneously, and seek to borrow strength across the different series to identify changepoints.  The nature of series that are tracked in applications, as well as the desire to evade to the greatest extent possible the curse of dimensionality, means that it is commonly assumed that the signal of interest is relatively sparse, in the sense that only a small proportion of the constituent series undergo a change.  Furthermore, the large majority of these works have focused on the retrospective (or \emph{offline}) challenges of detecting and estimating changes after seeing all of the available data \citep[e.g.][]{chan2015optimal,ChoFryzlewicz2014,Jirak2015,Cho2016,SohChandrasekaran2017,WangSamworth2018,EnikeevaHarchaoui2019,PYWR2019,kaul2021inference,LGS2019,londschien2021change,rinaldo2021localizing,FWS2021}.  Nevertheless, the related problem where one observes data sequentially and seeks to declare changes as soon as possible after they have occurred, is nowadays receiving increasing attention \citep[e.g.][]{kirch2019sequential, dette2020likelihood, gosmann2020sequential, yu2021optimal, CWS2021}.  Although the focus of our review here has been on recent developments, including finite-sample results in multivariate and high-dimensional settings, we also mention that changepoint analysis has a long history \citep[e.g.][]{Page1954}.  Entry points to this classical literature include \citet{CsorgoHorvath1997} and \citet{HorvathRice2014}.  For univariate data, sequential changepoint detection is also studied under the banner of statistical process control \citep{Duncan1952,TNB2014}.  In the field of high-dimensional statistical inference more generally, uncertainty quantification has become a major theme over the last decade, originating with influential work on the debiased Lasso in (generalized) linear models \citep{javanmard2014confidence,van2014asymptotically,zhang2014confidence}, and subsequently developed in other settings \citep[e.g.][]{jankova2015confidence,yu2018confidence}.
	
	The aim of this paper is to propose methods to address two new inferential challenges associated with the high-dimensional, sequential detection of a sparse change in mean.  The first is to provide a confidence interval for the location of the changepoint, while the second is to estimate the signal set of indices of coordinates that undergo the change.  Despite the importance of uncertainty quantification and signal support recovery in changepoint applications, neither of these problems has previously been studied in the multivariate sequential changepoint detection literature, to the best of our knowledge.  Of course, one option here would be to apply an offline confidence interval construction \citep[e.g.,][]{kaul2021inference} after a sequential procedure has declared a change.  However, this would be to ignore the essential challenge of the sequential nature of the problem, whereby one wishes to avoid storing all historical data, to enable inference to be carried out in an \emph{online} manner.  By this we mean that the computational complexity for processing each new observation, as well as the storage requirements, should depend only on the number of bits needed to represent the new data point observed\footnote{Here, we ignore the errors in rounding real numbers to machine precision; thus, we do not distinguish between continuous random variables and quantized versions where the data have been rounded to machine precision.}.  The online requirement turns out to impose severe restrictions on the class of algorithms available to the practitioner, and lies at the heart of the difficulty of the problem.
	
	To give a brief outline of our construction of a confidence interval with guaranteed $(1-\alpha)$-level coverage, consider for simplicity the univariate setting, where $(X_n)_{n \in \mathbb{N}}$ form a sequence of independent random variables with $X_1, \ldots, X_z \stackrel{\mathrm{iid}}{\sim} \mathcal{N}(0, 1)$ and $X_{z+1}, X_{z+2}, \ldots \stackrel{\mathrm{iid}}{\sim} \mathcal{N}(\theta, 1)$. Without loss of generality, we assume that $\theta > 0$.  Suppose that $\theta$ is known to be at least $b > 0$ and, for $n \in \mathbb{N}$, define \emph{residual tail lengths}
	\begin{equation}
	\label{Eq:OneDimTail}
	    	t_{n,b} := \argmax_{0 \leq h \leq n} \sum_{i=n-h+1}^{n} (X_i - b/2). 
	\end{equation}
	In the case of a tie, we choose the smallest $h$ achieving the maximum. Since $\sum_{i=n-h+1}^{n} (X_i - b/2)$ can be viewed as the likelihood ratio statistic for testing the null of $\mathcal{N}(0,1)$ against the alternative of $\mathcal{N}(b,1)$ using $X_{n-h+1}, \ldots, X_{n}$, the quantity $t_{n,b}$ is the tail length for which the likelihood ratio statistic is maximized.  If $N$ is the stopping time defining a good sequential changepoint detection procedure, then, intuitively, $N-t_{N,b}$ should be close to the true changepoint location $z$, and almost pivotal.  This motivates the construction of a confidence interval of the form $\bigl[\max\bigl\{N - t_{N,b} - g(\alpha,b), 0\bigr\}, N\bigr]$, where we control the tail probability of the distribution of $N-t_{N,b}$ to choose $g(\alpha,b)$ so as to ensure the desired coverage.  In the multivariate case, considerable care is required to handle the post-selection nature of the inferential problem, as well as to determine an appropriate left endpoint for the confidence interval.  For this latter purpose, we only assume a lower bound on the Euclidean norm of the vector of mean change, and employ a delicate multivariate and multiscale aggregation scheme; see Section~\ref{Sec:Method} for details, as well as Section~\ref{SubSec:Beta} for further discussion.
	
	The procedure for the inference tasks discussed above, which we call \texttt{ocd\_CI} (short for \textbf{o}nline \textbf{c}hangepoint \textbf{d}etection \textbf{C}onfidence \textbf{I}ntervals), can be run in conjunction with any base sequential changepoint detection procedure.  However, we recommend using the \texttt{ocd} algorithm introduced by \citet{CWS2021}, or its variant \texttt{ocd$'$}, which provides guarantees on both the average and worst-case detection delays, subject to a guarantee on the \emph{patience}, or average false alarm rate under the null hypothesis of no change.  Crucially, these are both online algorithms.  The corresponding inferential procedures inherit this same online property, thereby making them applicable even in very high-dimensional settings and where changes may be rare, so we may need to see many new data points before declaring a change.
	
	In Section~\ref{Sec:Theory} we study the theoretical performance of the \texttt{ocd\_CI} procedure. In particular, we prove in Theorem~\ref{Thm:Coverage} that, whenever the base sequential detection procedure satisfies a patience and detection delay condition, the confidence interval has at least nominal coverage for a suitable choice of input parameters.  Theorem~\ref{Thm:Length} provides a corresponding guarantee on the length of the interval. In Section~\ref{Sec:ocd_Base}, we show that by using \texttt{ocd}$'$ as the base procedure, the aforementioned patience and detection delay condition is indeed satisfied. As a result, the output confidence interval has guaranteed nominal coverage and the length of the interval is of the same order as the average detection delay for the base \texttt{ocd$'$} procedure, up to a poly-logarithmic factor.  This is remarkable in view of the intrinsic challenge that the better such a changepoint detection procedure performs, the fewer post-change observations are available for inferential tasks.   
	
	A very useful byproduct of our \texttt{ocd\_CI} methodology is that we obtain a natural estimate of the set of signal coordinates (i.e.~those that undergo change).  In Theorem~\ref{Thm:Support}, we prove that, with high probability, it is able both to recover the effective support of the signal (see Section~\ref{Sec:CovGuarantee} for a formal definition), and to avoid noise coordinates.  We then broaden the scope of applicability of our methodology in Section~\ref{SubSec:Relax} by relaxing our distributional assumptions to deal with sub-Gaussian or sub-exponential data.  Finally, in Section~\ref{SubSec:Beta}, we introduce a modification of our algorithm that permits an arbitrarily loose lower bound $\beta > 0$ on the Euclidean norm of the vector of mean change to be employed, with only a logarithmic increase in the confidence interval length guarantee and the computational~cost.
	
	An attraction of our theoretical results is that we are able to handle arbitrary spatial (cross-sectional) dependence between the different coordinates of our data stream.  On the other hand, two limitations of our analysis for practical use are that real data may exhibit both heavier than sub-exponential tails and temporal dependence.  While a full theoretical analysis of the \texttt{ocd\_CI} algorithm in these contexts appears to be challenging, we have made some practical suggestions regarding these issues in Sections~\ref{SubSec:Relax} and~\ref{Sec:USCovid} respectively. 
	
	
	Section~\ref{Sec:Simulation} is devoted to a study of the numerical performance of our methodological proposals.  Our simulations confirm that the \texttt{ocd\_CI} methodology (with the \texttt{ocd} base procedure) attains the desired coverage level across a wide range of parameter settings, that the average confidence interval length is of comparable order to the average detection delay and that our support recovery guarantees are validated empirically. We further demonstrate the way in which naive application of offline methods may lead to poor performance in this problem.  Moreover, in Section~\ref{Sec:USCovid}, we apply our methods to excess death data from the Covid-19 pandemic in the US.  Proofs, auxiliary results, extensions to sub-Gaussian and sub-exponential settings and additional simulation results are provided in the supplementary material (Section~\ref{Sec:Supp}).
	
	We conclude this introduction with some notation used throughout the paper.  We write $\mathbb{N}_0$ for the set of all non-negative integers. For $d \in \mathbb{N}$, we write $[d]:=\{1, \ldots, d\}$. Given $a, b \in \mathbb{R}$, we denote $a \vee b := \max(a, b)$ and $a \wedge b := \min(a, b)$. For a set $S$, we use $\mathbbm{1}_S$ and $|S|$ to denote its indicator function and cardinality respectively. For a real-valued function $f$ on a totally ordered set $S$, we write $\sargmax_{x \in S} f(x) := \min \argmax_{x\in S} f(x)$ and $\largmax_{x \in S} f(x) := \max \argmax_{x\in S} f(x)$ for the smallest and largest maximizers of $f$ in $S$, and define $\sargmin_{x \in S} f(x)$ and $\largmin_{x \in S} f(x)$ analogously.  For a vector $v = \bigl(v^1, \ldots, v^M\bigr)^\top \in \mathbb{R}^M$, we define $\|v\|_0 := \sum_{i=1}^M \mathbbm{1}_{\{v^i \neq 0\}}$, $\|v\|_2:=\bigl\{\sum_{i=1}^{M} (v^i)^2\bigr\}^{1/2}$ and $\|v\|_\infty := \max_{i \in [M]} |v^i|$. In addition, for $j \in [M]$, we define $\|v^{-j}\|_2 := \bigl\{ \sum_{i: i \neq j} (v^i)^2 \bigr\}^{1/2}$. For a matrix $A=(A^{i,j}) \in \mathbb{R}^{d_1\times d_2}$ and $j \in [d_2]$, we write $A^{\bdot, j} := \bigl(A^{1, j}, \ldots, A^{d_1, j}\bigr)^\top \in \mathbb{R}^{d_1}$ and $A^{-j, j} := \bigl(A^{1, j}, \ldots, A^{j-1, j}, A^{j+1, j}\ldots, A^{d_1, j}\bigr)^\top \in \mathbb{R}^{d_1 - 1}$. We use $\Phi(\cdot)$, $\bar{\Phi}(\cdot)$ and $\phi(\cdot)$ to denote the distribution function, survivor function and density function of the standard normal distribution respectively. For two real-valued random variables $U$ and~$V$, we write $U \geq_{\mathrm{st}} V$ or $V \leq_{\mathrm{st}} U$  if $\mathbb{P}(U \leq x) \leq \mathbb{P}(V \leq x)$ for all $x\in \mathbb{R}$. We adopt conventions that an empty sum is $0$ and that $\min \emptyset := \infty$, $\max \emptyset := - \infty$.

	\section{Confidence interval construction and support estimation methodology} \label{Sec:Method}
	
	In the multivariate sequential changepoint detection problem, we observe $p$-variate observations $X_1,X_2,\ldots$ in turn, and seek to report a stopping time $N$ by which we believe a change has occurred. Here and throughout, a stopping time is understood to be with respect to the natural filtration, so that the event $\{N = n\}$ belongs to the $\sigma$-algebra generated by $X_1,\ldots,X_n$. The focus of this work is on changes in the mean of the underlying process, and we denote the time of the changepoint by $z$.  Moreover, since our primary interest is in high-dimensional settings, we will also seek to exploit sparsity in the vector of mean change.  Given $\alpha \in (0,1)$, then, our primary goal is to construct a confidence interval $\mathcal{C} \equiv \mathcal{C}(X_1,\ldots,X_N,\alpha)$ with the property that $z \in \mathcal{C}$ with probability at least $1-\alpha$.
	
	For $i \in \mathbb{N}$ and $j \in [p]$, let $X_i^j$ denote the $j$th coordinate of $X_i$. The \texttt{ocd\_CI} algorithm relies on a lower bound $\beta > 0$ for the $\ell_2$-norm of the vector of mean change, sets of signed scales $\mathcal{B}$ and $\mathcal{B}_0$ defined in terms of $\beta$ and a base sequential changepoint detection procedure CP.  As CP processes each new data vector, we update the matrix of residual tail lengths
	$(t_{n,b}^j)_{j \in [p],b \in \mathcal{B} \cup \mathcal{B}_0}$ with $t_{n,b}^{j} := \sargmax_{0 \leq h \leq n} \sum_{i=n-h+1}^n (X_i^{j} - b/2)$, as well as corresponding tail partial sum vectors $\bigl(A_{n,b}^{\cdot,j}\bigr)_{j \in [p],b \in \mathcal{B} \cup \mathcal{B}_0}$, where $A_{n,b}^{j',j} := \sum_{i=n-t_{n,b}^j+1}^n X_i^{j'}$. 
	
	After the base procedure CP declares a change at a stopping time $N$, we identify an ``anchor'' coordinate $\hat{j} \in [p]$ and a signed anchor scale $\hat{b} \in \mathcal{B}$, where 
	\[
	(\hat{j}, \hat{b}) := \argmax_{(j,b)\in[p]\times \mathcal{B}} \sum_{j' \in [p] \setminus \{j\}}\frac{  \bigl(A_{N,b}^{j',j}\bigr)^2 }{t_{N,b}^j \vee 1} \mathbbm{1}_{\bigl\{|A_{N,b}^{j',j}| \geq a\sqrt{t_{N,b}^j \vee 1}\bigr\}}.
	\]
	The intuition is that the anchor coordinate and signed anchor scale are chosen so that the final $t_{N,\hat{b}}^{\hat{j}}$ observations provide the best evidence among all of the residual tail lengths against the null hypothesis of no change.  Meanwhile,  $A_{N,\hat{b}}^{\cdot,\hat{j}}$ aggregates the last $t_{N,\hat{b}}^{\hat{j}}$ observations in each coordinate, providing a measure of the strength of this evidence against the null.
	
	The main idea of our confidence interval construction is to seek to identify coordinates with large post-change signal.  To this end, observe when $t_{N, \hat{b}}^{\hat{j}}$ is not too much larger than $N-z$, the quantity $E_{N, \hat{b}}^{j,\hat{j}} := A_{N, {\hat{b}}}^{j,\hat{j}} / (t_{N, {\hat{b}}}^{\hat{j}}\vee 1)^{1/2}$ should be centered close to  $\theta^j(t_{N, {\hat{b}}}^{\hat{j}})^{1/2}$ for $j \in [p] \setminus \{\hat{j}\}$, with variance close to 1.  Indeed, if $\hat{j}$, $\hat{b}$, $N$ and $t_{N,\hat{b}}^{\hat{j}}$ were fixed, and if $0 < t_{N,\hat{b}}^{\hat{j}} \leq N-z$, then the former quantity would have unit variance around this centering value.  The random nature of these quantities, however, introduces a post-selection inference aspect to the problem.  Nevertheless, by choosing an appropriate threshold value $d_1 > 0$, we can ensure that with high probability, when $j \neq \hat{j}$ is a noise coordinate, we have $|E_{N, \hat b}^{j, \hat j}| < d_1$, and when $j \neq \hat{j}$ is a coordinate with sufficiently large signal, there exists a signed scale $b\in (\mathcal{B}\cup \mathcal{B}_0) \cap [-|\theta^j|,|\theta^j|]$, having the same sign as $\theta^j$, for which $\bigl|E_{N, \hat b}^{j, \hat j}\bigr| - |b|(t_{N, {\hat{b}}}^{\hat{j}})^{1/2} \geq d_1$.  In fact, such a signed scale, if it exists, can always be chosen to be from $\mathcal{B}_0$.  As a convenient byproduct, the set of indices $j$ for which the latter inequality holds, which we denote as $\hat{\mathcal{S}}$, forms a natural estimate of the set of coordinates in which the mean change is large.
	
	For each $j \in \hat{\mathcal{S}}$, there exists a largest scale $b \in (\mathcal{B} \cup \mathcal{B}_0) \cap (0,\infty)$ for which $\bigl|E_{N, \hat b}^{j, \hat j}\bigr| - b(t_{N, {\hat{b}}}^{\hat{j}})^{1/2} \geq d_1$.  We denote the signed version of this quantity, where the sign is chosen to agree with that of $E_{N,\hat{b}}^{j,\hat{j}}$, by $\tilde{b}^j$; this can be regarded as a shrunken estimate of $\theta^j$, so plays the role of the lower bound $b$ from the univariate problem discussed in the introduction.  Finally, then, our confidence interval is constructed as the intersection over indices $j \in \hat{\mathcal{S}}$ of the confidence interval from the univariate problem in coordinate $j$, with signed scale $\tilde{b}^j$.  
	
  As a device to facilitate our theory, the \texttt{ocd\_CI}  algorithm allows the practitioner the possibility of observing a further $\ell$ observations after the time of changepoint declaration, before constructing the confidence interval.  The additional observations are used to determine the anchor coordinate $\hat{j}$ and scale $\hat{b}$, as well as the estimated support $\hat{\mathcal{S}}$ and the estimated scale $\tilde{b}^j$ for each $j \in \hat{\mathcal{S}}$.  Thus, the extra sampling is used to guard against an unusually early changepoint declaration that leaves very few post-change observations for inference.  Nevertheless, we will see in Theorem~\ref{Thm:Coverage} below that the output confidence interval has guaranteed nominal coverage even with $\ell=0$, so that additional observations are only used to control the length of the interval.  In fact, even for this latter aspect, the numerical evidence presented in Section~\ref{Sec:Simulation} indicates that $\ell=0$ provides confidence intervals of reasonable length in practice.  Similarly, Theorem~\ref{Thm:Support} ensures that with high probability, our support estimate $\hat{\mathcal{S}}$ contains no noise coordinates (i.e.~has false positive control) even with $\ell = 0$, so that the extra sampling is only used to provide false negative control.
	
	Pseudo-code for this \texttt{ocd\_CI} confidence interval construction is given in Algorithm~\ref{Alg:CI}, where we suppress the $n$ dependence on quantities that are updated at each time step.  The computational complexity per new observation, as well as the storage requirements, of this algorithm is equal to the sum of the corresponding quantities for the CP base procedure and $O\bigl(p^2\log(ep)\bigr)$ regardless of the observation history.  Thus the \texttt{ocd\_CI} method inherits the property of being an online algorithm, as discussed in the introduction, from any online CP base procedure. 

 \begin{algorithm}[htbp!]
		\KwIn{$X_1, X_2, \ldots \in \mathbb{R}^p$ observed sequentially, $\beta>0$, $a\geq 0$, an online changepoint detection procedure $\mathrm{CP}$, $d_1, d_2 > 0$ and $\ell \in \mathbb{N}_0$}
		\KwSet{$b_{\min}= \frac{\beta}{\sqrt{2^{\lfloor \log_2 (2p)\rfloor} \log_2(2p)}}$, $\mathcal{B}_0 = \{\pm b_{\min}\}$, 
		$\mathcal{B} = \bigl\{\pm 2^{m/2} b_{\min}: m = 1,\ldots,\lfloor \log_2 (2p) \rfloor \bigr\}$,  $n = 0$, $A_b = \mathbf{0} \in \mathbb{R}^{p\times p}$ and $t_b = 0\in\mathbb{R}^p$ for all $b \in \mathcal{B}\cup \mathcal{B}_0$}
		\Repeat{$\mathrm{CP}$ declares a change}{
			$n \leftarrow n+1$\\
			observe new data vector $X_n$ and update $\mathrm{CP}$ with $X_n$ \\
			\For{$(j, b) \in [p]\times (\mathcal{B} \cup \mathcal{B}_0)$}{
				$t^j_b \leftarrow t^j_b +1$ \\ 
				$A^{\bdot,j}_b \leftarrow A^{\bdot,j}_b + X_n$ \\
				\If{$bA^{j,j}_{b} - b^2t^j_b/2 \leq 0$}{                 
					$t^j_b \leftarrow 0$ and  $A^{\bdot,j}_b \leftarrow 0$
				} 
			}
		}
		Observe $\ell$ new data vectors $X_{n+1}, \ldots, X_{n+\ell}$ \\
		Set $E_{b}^{j', j} \leftarrow \frac{  A_{b}^{j',j}+\sum_{i=n+1}^{n+\ell} X_i^{j'} }{\sqrt{(t_{b}^j +\ell)\vee 1}}$ for $j', j\in[p]$, $b\in\mathcal{B}\cup\mathcal{B}_0$ \\
		
		Compute $Q^j_b \leftarrow \sum_{j' \in [p] \setminus \{j\}}(E_b^{j',j})^2 \mathbbm{1}_{\{|E_b^{j',j}| \geq a\}}$ for $j \in [p]$, $b \in \mathcal{B}$ 
		$(\hat{j}, \hat{b}) \leftarrow \argmax_{(j,b)\in[p]\times \mathcal{B}} Q_{b}^j$ \\
		$\hat{\mathcal{S}} \leftarrow \Bigl\{j \in [p]\setminus\{\hat{j}\}: \bigl|E_{\hat{b}}^{j,\hat{j}}\bigr| - b_{\min} (t_{\hat{b}}^{\hat{j}} + \ell)^{1/2}  \geq  d_1 \Bigr\}$ \\
		\For{$j \in \hat{\mathcal{S}}$}{$\tilde{b}^j \leftarrow \mathrm{sgn}\bigl(E_{\hat{b}}^{j, \hat{j}}\bigr)\max\Bigl\{b \in (\mathcal{B}\cup\mathcal{B}_0) \cap (0, \infty): \bigl|E_{\hat{b}}^{j,\hat{j}}\bigr| - b(t_{\hat{b}}^{\hat{j}} + \ell)^{1/2}  \geq d_1  \Bigr\}$}
		\KwOut{Confidence interval $\mathcal{C}=\Bigl[ \max \Bigl\{n-\min_{j \in \hat{\mathcal{S}}} \Bigl\{t_{\tilde{b}^j}^j + \frac{d_2}{(\tilde{b}^j)^2} \Bigr\}, 0\Bigr\}, n\Bigr]$}
		\caption{Pseudo-code for the confidence interval construction algorithm \texttt{ocd\_CI}}
		\label{Alg:CI}
	\end{algorithm}
	
	A natural choice for the base online changepoint detection procedure CP is the \texttt{ocd} algorithm, or its variant \texttt{ocd}$'$, introduced by \citet{CWS2021}. Both are online algorithms, with computational complexity per new observation and storage requirements of $O\bigl(p^2\log(ep)\bigr)$. The \texttt{ocd}$'$ base procedure is considered for the theoretical analysis in Section~\ref{Sec:Theory} due to its known patience and detection delay guarantees, while we prefer \texttt{ocd} for numerical studies and practical use. For the reader's convenience, the \texttt{ocd} and \texttt{ocd}$'$ algorithms are provided as Algorithms~\ref{Alg:ocd} and~\ref{Alg:ocd_variant} respectively in Section~\ref{Sec:ocdprimealgorithm}. 
	
	\section{Theoretical analysis} \label{Sec:Theory}
    
    Throughout this section, we will assume that the sequential observations $X_1,X_2,\ldots$ are independent, and that for some unknown covariance matrix $\Sigma \in \mathbb{R}^{p \times p}$ whose diagonal entries are all equal to $1$, there exist $z \in \mathbb{N}_0$ and $\theta = (\theta^1,\ldots,\theta^p)^\top \neq 0$ for which $X_1,\ldots,X_z \sim \mathcal{N}_p(0,\Sigma)$ and $X_{z+1},X_{z+2},\ldots \sim \mathcal{N}_p(\theta,\Sigma)$.  We let $\vartheta := \|\theta\|_2$, and write $\mathbb{P}_{z,\theta,\Sigma}$ for probabilities computed under this model, though in places we omit the subscripts for brevity.  Define the \emph{effective sparsity} of~$\theta$, denoted $s(\theta)$, to be the smallest $s \in \bigl\{2^0, 2^1, \ldots, 2^{\lfloor\log_2(p)\rfloor}\bigr\}$ such that the corresponding \emph{effective support} $\mathcal{S}(\theta) := \bigl\{j \in [p]: |\theta^j| \geq \|\theta\|_2/\sqrt{s \log_2(2p)}\bigr\}$ has cardinality at least $s(\theta)$.  Thus, the sum of squares of coordinates in the effective support of $\theta$ has the same order of magnitude as $\|\theta\|_2^2$, up to logarithmic factors.  Moreover, if  at most $s$ components of~$\theta$ are non-zero, then $s(\theta) \leq s$, and the equality is attained when, for example, all non-zero coordinates have the same magnitude.
    
    For $r > 0$ and an online changepoint detection procedure $\mathrm{CP}$ characterized by an extended stopping time $N$, we define 
	\begin{equation}
	\label{Eq:grn}
	g(r; N):= \sup_{z \in \mathbb{N}_0} \mathbb{P}_{z,\theta,\Sigma}(N > z+r).
	\end{equation}
	\subsection{Coverage Probability and Length of the Confidence Interval}
	\label{Sec:CovGuarantee}
	
	The following theorem shows that the confidence interval constructed in the \texttt{ocd\_CI} algorithm has the desired coverage level whenever the base online changepoint detection procedure satisfies a patience and detection delay condition.
	\begin{thm} \label{Thm:Coverage}
	Let $p \geq 2$ and fix $\alpha \in (0, 1)$. Suppose that $\vartheta \geq \beta > 0$. Let $\mathrm{CP}$ be an online changepoint procedure characterized by an extended stopping time $N$ satisfying 
		\begin{equation} \label{Eq:detectassumption}
		    \mathbb{P}_{z,\theta,\Sigma}(N \leq z) + g(r; N) + 4rp^2\log_2^2(4p) e^{-r\beta^2/(8s\log_2(2p))} \leq \frac{3}{4} \alpha
		\end{equation}
		for some $r \geq 1$.  Then, with inputs $(X_t)_{t\in \mathbb{N}}$, $\beta > 0$, $a\geq 0$, $\mathrm{CP}$, $d_1 = \sqrt{\frac{5r\beta^2}{9s\log_2(2p)}}$, $d_2 = 4d_1^2$ and $\ell \geq 0$, the output confidence interval $\mathcal{C}$ from Algorithm~\ref{Alg:CI} satisfies
		\begin{align*}
			\mathbb{P}_{z,\theta,\Sigma}( z \in \mathcal{C} ) \geq 1-\alpha.
		\end{align*}
	\end{thm}
	As mentioned in Section~\ref{Sec:Method}, our coverage guarantee in Theorem~\ref{Thm:Coverage} holds even with $\ell = 0$, i.e.~with no additional sampling.  Condition~\eqref{Eq:detectassumption} places a joint assumption on the base changepoint procedure $\mathrm{CP}$ and the parameter $r$, the latter of which appears in the inputs $d_1$ and $d_2$ of Algorithm~\ref{Alg:CI}. The first term on the left-hand side of~\eqref{Eq:detectassumption} is the false alarm rate of the stopping time $N$ associated with $\mathrm{CP}$.  The second term can be interpreted as an upper bound on the probability of the detection delay of $N$ being larger than $r$, and in addition we also need $r$ to be at least of order $s/\beta^2$ up to logarithmic factors for the third term to be sufficiently small.  See Section~\ref{Sec:ocd_Base} for further discussion, where in particular we provide a choice of $r$ for which~\eqref{Eq:detectassumption} holds with the \texttt{ocd}$'$ base procedure.
	
	We now provide a guarantee on the length of the \texttt{ocd\_CI} confidence interval.
	\begin{thm} \label{Thm:Length}
		Fix $\alpha \in (0, 1)$. Assume that $\theta$ has an effective sparsity of $s:= s(\theta) \geq 2$ and that $\vartheta \geq \beta > 0$. let $\mathrm{CP}$ be an online changepoint detection procedure characterized by an extended stopping time $N$ that satisfies~\eqref{Eq:detectassumption} for some $r \geq 1$.  Then there exists a universal constant $C > 0$ such that, with inputs $(X_t)_{t\in \mathbb{N}}$, $\beta > 0$, $a = C\sqrt{\log(rp/\alpha)}$, $\mathrm{CP}$, $d_1 = \sqrt{\frac{5r\beta^2}{9s\log_2(2p)}}$, $d_2 = 4d_1^2$, $\ell \geq 80r$, the length $L$ of the output confidence interval $\mathcal{C}$ satisfies
		\[
		\mathbb{P}_{z,\theta,\Sigma}(L > 8r) \leq \alpha.
		\]
	\end{thm}
	As mentioned following Theorem~\ref{Thm:Coverage}, we can take $r$ to be the maximum of an appropriate quantile of the detection delay distribution of $\mathrm{CP}$ and a quantity that is of order $s/\beta^2$ up to logarithmic factors.  The main conclusion of Theorem~\ref{Thm:Length} is that, with high probability, the length of the confidence interval is then of this same order $r$.  Whenever the quantile of the detection delay distribution achieves the maximum above --- which is the case, up to logarithmic factors, for the \texttt{ocd}$'$ base procedure (see Proposition~\ref{prop:ocdsatisfyassumption}) --- we can conclude that with high probability, the length of the \texttt{ocd\_CI} confidence interval is of the same order as this detection delay quantile (which is the best one could hope for).   Note that the choices of inputs in Theorem~\ref{Thm:Length} are identical to those in Theorem~\ref{Thm:Coverage}, except that we now ask for order $r$ additional observations after the changepoint declaration.
	
	\subsection{Support Recovery}
	\label{SubSec:Support}
	Recall the definition of $\mathcal{S}(\theta)$ from the beginning of this section, and denote $\mathcal{S}_{\beta}(\theta) := \bigl\{j \in [p]: |\theta^j| \geq b_{\min} \bigr\}$, where $b_{\min}$, defined in Algorithm~\ref{Alg:CI}, is the smallest positive scale in $\mathcal{B}\cup\mathcal{B}_0$, We will suppress the dependence on $\theta$ of both these quantities in this subsection.  Theorem~\ref{Thm:Support} below provides a support recovery guarantee for $\hat{\mathcal{S}}$, defined in Algorithm~\ref{Alg:CI}.  Since neither $\hat{\mathcal{S}}$ nor the anchor coordinate $\hat{j}$ defined in the algorithm depend on $d_2$, we omit its specification; the choices of other tuning parameters mimic those in Theorems~\ref{Thm:Coverage} and~\ref{Thm:Length}.
	\begin{thm} \label{Thm:Support}
	Let $p \geq 2$ and fix $\alpha \in (0,1)$. Suppose that $\vartheta \geq \beta > 0$. Let $\mathrm{CP}$ be an online changepoint detection procedure characterized by an extended stopping time $N$ that satisfies~\eqref{Eq:detectassumption} for some $r \geq 1$. 
	
	\medskip
		\noindent (a) Then, with inputs $(X_t)_{t\in \mathbb{N}}$, $\beta > 0$, $a\geq 0$, $\mathrm{CP}$, $d_1 = \sqrt{\frac{5r\beta^2}{9s\log_2(2p)}}$, $\ell \geq 0$, we have
		\[
		\mathbb{P}_{z,\theta,\Sigma}(\hat{\mathcal{S}} \subseteq \mathcal{S}_{\beta}) \geq 1-\alpha.
		\]
		(b) Now assume that $\theta$ has an effective sparsity of
		$s:= s(\theta) \geq 2$. Then there exists a universal constant $C > 0$ such that, with inputs $(X_t)_{t\in \mathbb{N}}$, $\beta > 0$, $a = C\sqrt{\log(rp/\alpha)}$, $\mathrm{CP}$, $d_1 = \sqrt{\frac{5r\beta^2}{9s\log_2(2p)}}$, $\ell \geq 80r$, we have
		\[
		\mathbb{P}_{z,\theta,\Sigma}(\hat{\mathcal{S}}\cup \{\hat{j}\} \supseteq  \mathcal{S}) \geq 1-\alpha.
		\]
	\end{thm}
	Note that $\mathcal{S} \subseteq \mathcal{S}_{\beta} \subseteq \{j \in [p]:\theta^j\neq 0\}$.  Thus, part~(a) of the theorem reveals that with high probability, our support estimate $\hat{\mathcal{S}}$ does not contain any noise coordinates.  Part~(b) offers a complementary guarantee on the inclusion of all ``big'' signal coordinates, provided we augment our support estimate with the anchor coordinate $\hat{j}$.  See also the further discussion of this result following Proposition~\ref{Prop:LowerBound} below and in Section~\ref{Sec:ocd_Base}.
	
We now turn our attention to the optimality of our support recovery algorithm, by establishing a complementary minimax lower bound on the performance of any support estimator.  In fact, we can already establish this optimality by restricting the cross-sectional covariance matrix to be the identity matrix.  Thus, given $\theta\in\mathbb{R}^p$ and $z\in\mathbb{N}_0$, we write $\mathbb{P}_{z,\theta}$ for a probability measure under which $(X_n)_{n \in \mathbb{N}}$ are independent with $X_n \sim \mathcal{N}_p(\theta\mathbbm{1}_{\{n > z\}}, I_p)$.  For $r > 0$ and $m \in [p] \cup \{0\}$, write 
\[
\Theta_{r,m} := \bigl\{\theta \in \mathbb{R}^p: |\{j \in [p]:|\theta^j|\leq 1/(8\sqrt{r})\}| \geq m\bigr\}.
\]
Define $\mathcal{T}$ to be the set of stopping times with respect to the natural filtration $(\mathcal{F}_n)_{n \in \mathbb{N}_0}$, and set
\[
	\mathcal{T}_{r,m}:= \biggl\{N \in \mathcal{T}: \sup_{z\in\mathbb{N} \cup \{0\}, \theta\in \Theta_{r,m}}\mathbb{P}_{z,\theta}(N > z+r) \leq \frac{1}{4}\biggr\}.
\]
Write $2^{[p]}$ for the power set of $[p]$, equipped with the  symmetric difference metric $d: (A, B) \mapsto |(A\setminus B)\cup(B\setminus A)|$. For any stopping time $N$, denote
\[	
\mathcal{J}_N := \{\psi: (\mathbb{R}^p)^{\infty} \rightarrow 2^{[p]}: \text{$\psi$ is $\mathcal{F}_N$-measurable}\},
	\]
where we recall that $\psi$ is said to be $\mathcal{F}_N$-measurable if for any $A \in 2^{[p]}$ and $n \in \mathbb{N}_0$, we have that $\psi^{-1}(A) \cap \{N = n\}$ is $\mathcal{F}_n$-measurable.

\begin{prop}
For $r > 0$ and $m\geq 15$, we have 
\label{Prop:LowerBound}
    \[
    \inf_{N\in\mathcal{T}_{r,m}} \inf_{\psi\in \mathcal{J}_N}  \sup_{z\in\mathbb{N}_0, \theta\in \Theta_{r,m}} \mathbb{E}_{z,\theta} \,d\bigl(\psi(X_1,X_2,\ldots), \mathrm{supp}(\theta)\bigr) \geq \frac{m}{32}.
    \]
\end{prop}
This proposition considers any support estimation algorithm obtained from a stopping time in $\mathcal{T}_{r,m}$, and we note that such a competing procedure is even allowed to store all data up to this stopping time, in contrast to our online algorithm. This result can be interpreted as an optimality guarantee for the support recovery property of the \texttt{ocd\_CI} algorithm presented in Theorem~\ref{Thm:Support}(b), provided that the base procedure $N$ belongs to the class $\mathcal{T}_{r,m}$, and that $N$ and $r$ satisfy~\eqref{Eq:detectassumption}.  See Section~\ref{Sec:ocd_Base} below for further discussion.

\subsection{Using \texttt{ocd}\texorpdfstring{$'$}{'} as the Base Procedure}
\label{Sec:ocd_Base}
	
	In this subsection, we provide a value of $r$ that suffices for condition~\eqref{Eq:detectassumption} to hold when we take our base procedure to be \texttt{ocd}$'$.  For the convenience of the reader, this algorithm is presented as Algorithm~\ref{Alg:ocd_variant} in Section~\ref{Sec:ocdprimealgorithm}, where we also provide interpretation to the input parameters $\tilde{a}$, $T^{\mathrm{diag}}$ and $T^{\mathrm{off}}$.  
	
	\begin{prop}
	\label{prop:ocdsatisfyassumption}
	Fix $\alpha \in (0,1)$ and $\gamma > 0$. Assume that $\theta$ has an effective sparsity of $s := s(\theta) \geq 2$, that $\vartheta \geq \beta > 0$ and that $z \leq 2\alpha \gamma$.  Then with inputs $(X_t)_{t\in \mathbb{N}}$, $\beta > 0$, $\tilde{a} = \sqrt{2\log \{16p^2\gamma\log_2(2p)\}}$, $T^{\mathrm{diag}} = \log\{16p\gamma\log_2(4p)\}$ and $T^{\mathrm{off}} = 8\log\{16p\gamma\log_2(2p)\}$ in the $\mathtt{ocd}'$ procedure, there exists a universal constant $C' > 0$ such that for~all
	\begin{equation}
	    \label{Eq:r1}
	    r \geq \frac{C's\log_2(2p)\log\{p\gamma\alpha^{-1}(\beta^{-2}\vee 1)\}}{\beta^2} + 2 =: r_1, 
	\end{equation}
	the output $N$ satisfies~\eqref{Eq:detectassumption}.
	\end{prop}
	By combining Proposition~\ref{prop:ocdsatisfyassumption} with Theorems~\ref{Thm:Coverage},~\ref{Thm:Length} and~\ref{Thm:Support} respectively, we immediately arrive at the following corollary.
	\begin{cor}
	\label{cor:allresultsocd}
	Fix $\alpha \in (0, 1)$, $\gamma > 0$. Assume that $\theta$ has an effective sparsity of $s := s(\theta) \geq 2$, that $\vartheta \geq \beta > 0$ and that $z \leq 2\alpha \gamma$. Let $(X_t)_{t\in \mathbb{N}}$, $\beta > 0$, $\tilde{a} = \sqrt{2\log \{16p^2\gamma\log_2(2p)\}}$, $T^{\mathrm{diag}} = \log\{16p\gamma\log_2(4p)\}$ and $T^{\mathrm{off}} = 8\log\{16p\gamma\log_2(2p)\}$ be the inputs of the $\mathtt{ocd}'$ procedure.  Then the following statements hold:
	
	\medskip
	\noindent (a) With extra inputs $a \geq 0$, $\mathrm{CP} = \mathtt{ocd}'$, $d_1 = \sqrt{\frac{5r_1\beta^2}{9s\log_2(2p)}}$, $d_2 = 4d_1^2$ and $\ell \geq 0$ for Algorithm~\ref{Alg:CI}, the output confidence interval $\mathcal{C}$ and the support estimate $\hat{\mathcal{S}}$ satisfy $\mathbb{P}_{z,\theta,\Sigma}( z \in \mathcal{C} ) \geq 1-\alpha$ and $\mathbb{P}_{z,\theta,\Sigma}(\hat{\mathcal{S}} \subseteq \mathcal{S}_{\beta}) \geq 1-\alpha$.
	
	\medskip
	\noindent (b) There exists a universal constant $C > 0$ such that, with extra inputs $a = C\sqrt{\log(r_1p/\alpha)}$, $\mathrm{CP} = \mathtt{ocd}'$, $d_1 = \sqrt{\frac{5r_1\beta^2}{9s\log_2(2p)}}$, $d_2 = 4d_1^2$, $\ell \geq 80r_1$ for Algorithm~\ref{Alg:CI}, the length $L$ of the output confidence interval $\mathcal{C}$ and the support estimate satisfy $\mathbb{P}_{z,\theta,\Sigma}(L > 8r_1) \leq \alpha$ and $\mathbb{P}_{z,\theta,\Sigma}(\hat{\mathcal{S}}\cup \{\hat{j}\} \supseteq  \mathcal{S}) \geq 1-\alpha$.
	\end{cor}
	Corollary~\ref{cor:allresultsocd} reveals that, when  \texttt{ocd}$'$ is used as the base procedure, the \texttt{ocd\_CI} methodology provides guaranteed confidence interval coverage.  Moreover, up to poly-logarithmic factors, with an additional $O\bigl(1 \vee (s/\beta^2)\bigr)$ post-change observations, the \texttt{ocd\_CI} interval length is of the same order as the average detection delay. 	In terms of signal recovery, Corollary~\ref{cor:allresultsocd}(b) shows that with high probability, \texttt{ocd\_CI}  with inputs as given in that result selects all signal coordinates whose magnitude exceeds $\vartheta/s^{1/2}$, up to logarithmic factors.  Focusing on the case $\beta = \vartheta$ and where $s/\vartheta^2$ is bounded away from zero for simplicity of discussion (though see also Section~\ref{SubSec:Beta} for discussion of the effect of the choice of $\beta$), Proposition~\ref{prop:ocd'_tail} also reveals that the \texttt{ocd}$'$ base procedure belongs to $\mathcal{T}_{r,m}$ with $r$ of order $s/\vartheta^2$, up to logarithmic factors, and $m = |\{j:|\theta^j| \leq 1/(8\sqrt{r})\}|$.  On the other hand, Proposition~\ref{Prop:LowerBound} shows that any such support estimation algorithm makes on average a non-vanishing fraction of errors in distinguishing between noise coordinates and signals that are below the level $\vartheta/s^{1/2}$, again up to logarithmic factors.  In other words, with high probability, the \texttt{ocd\_CI} algorithm with base procedure \texttt{ocd}$'$ selects all signals that are strong enough (up to logarithmic factors) to be reliably detected, while at the same time including no noise coordinates (see Corollary~\ref{cor:allresultsocd}(a)). 
	
	\subsection{Relaxing the Gaussianity assumption}
	\label{SubSec:Relax}
	
	It is natural to ask to what extent the theory of Sections~\ref{Sec:CovGuarantee},~\ref{SubSec:Support} and~\ref{Sec:ocd_Base} can be generalised beyond the Gaussian setting.  The purpose of this subsection, then, is to describe how our earlier results can be modified to handle both sub-Gaussian and sub-exponential data.  Recall that we say a random variable $Z$ with $\mathbb{E}Z = 0$ is \emph{sub-Gaussian} with variance parameter $\sigma^2 > 0$ if $\mathbb{E}e^{\lambda Z} \leq e^{\sigma^2\lambda^2/2}$ for all $\lambda \in \mathbb{R}$, and is \emph{sub-exponential} with variance parameter $\sigma^2 > 0$ and rate parameter $A > 0$ if $\mathbb{E}e^{\lambda Z} \leq e^{\sigma^2\lambda^2/2}$ for all $|\lambda| \leq A$.
	
	We first consider the sub-Gaussian setting where $X_1,\ldots,X_z,X_{z+1} - \theta,X_{z+2} - \theta,\ldots$ are independent, each having sub-Gaussian components with variance parameter 1.  Note that this data generating mechanism no longer requires all pre-change observations to be identically distributed, and likewise the post-change observations need not all have the same distribution.  We assume that the base changepoint procedure, characterized by an extended stopping time $N$, satisfies a slightly strengthened version of~\eqref{Eq:detectassumption}, namely that
	\begin{equation}
	\label{Eq:strongerassumption}
	\mathbb{P}_{z,\theta,\Sigma}(N \leq z) + g(r; N) + 100rp^2\log_2^3(4p) (p\beta^{-2}  \vee 1 ) e^{-r\beta^2/(8s\log_2(2p))} \leq \frac{3}{4} \alpha.
	\end{equation}
	for some $r \geq 1$.  Under~\eqref{Eq:strongerassumption}, Theorems~\ref{Thm:Coverage},~\ref{Thm:Length} and~\ref{Thm:Support} hold with the same choices of input parameters.  Moreover, the \texttt{ocd}$'$ base procedure satisfies the conclusion of Proposition~\ref{prop:ocdsatisfyassumption}, i.e.~there exists a universal constant $C' > 0$ such that~\eqref{Eq:strongerassumption} holds for $r \geq r_1 = r_1(C')$ in~\eqref{Eq:r1}, provided that we use the modified input $\tilde{a} = \sqrt{2\log\bigl\{32p^2\gamma \log_2(2p)\bigr\}}$.
	
	Generalising these ideas further, now consider the model where $X_1,\ldots,X_z,X_{z+1} - \theta,X_{z+2} - \theta,\ldots$ are independent, each having sub-exponential components with variance parameter 1 and rate parameter $A > 0$.  In this setting, provided the base procedure satisfies~\eqref{Eq:strongerassumption} for some $r \geq 1$ and $\vartheta \leq \sqrt{2A^2 \log_2(2p)}$,  Theorems~\ref{Thm:Coverage},~\ref{Thm:Length} and~\ref{Thm:Support} hold when we redefine $a := C\max\bigl\{\sqrt{\log(rp/\alpha)}, \frac{1}{A}\log(rp/\alpha)\bigr\}$ and $d_1 := \max\Bigl\{ \sqrt\frac{5r\beta^2}{9s\log_2(2p)}, \frac{5r\beta^2}{9As\log_2(2p)}\Bigr\}$.  Furthermore, 
	with the modified input $\tilde{a} = \sqrt{2\log\bigl\{32p^2\gamma \log_2(2p) \bigr\} } \vee \frac{2}{A} \log\bigl\{32p^2\gamma \log_2(2p) \bigr\}$, the \texttt{ocd}$'$ base procedure satisfies the conclusion of  Proposition~\ref{prop:ocdsatisfyassumption} for 
	\[ 
	r \geq  \frac{C's\log_2(2p)\log^2\{p\gamma\alpha^{-1}(\beta^{-2}\vee 1)\}\max\{1, A^{-2} \log(p\gamma) \} }{\beta^2} + 2,
	\]
	where $C' > 0$ is a universal constant.
	
	
	The claims made in the previous two paragraphs are justified in Section~\ref{Sec:BeyondGaussianity}.  These results confirm the flexible scope of the \texttt{ocd\_CI} methodology beyond the original Gaussian setting, at least as far as sub-exponential tails are concerned.  Where data may exhibit heavier tails than this, clipping (truncation) and quantile transformations may represent viable ways to proceed, though further research is needed to confirm the theoretical validity of such approaches.

\subsection{Confidence interval construction with unknown signal size}
\label{SubSec:Beta}

In some settings, an experienced practitioner may have a reasonable idea of the magnitude~$\vartheta$ of the $\ell_2$-norm of the vector of mean change that would be of interest to them, and this would facilitate a choice of a lower bound $\beta$ for $\vartheta$ in Algorithm~\ref{Alg:CI}.  However, it is also worth considering the effect of this choice, and the extent to which its impact can be mitigated.  

We first remark that the coverage probability guarantee for the \texttt{ocd\_CI} interval in Corollary~\ref{cor:allresultsocd} remains valid for any (arbitrarily loose) lower bound $\beta$ on $\vartheta$.  The only issue is in terms of power: if $\beta$ is chosen to be too small, then both the average detection delay and the high-probability bound on the length of the confidence interval may be inflated.  In the remainder of this subsection, then, we describe a simple modification to Algorithm~\ref{Alg:CI} that permits a loose lower bound $\beta$ to be employed that retains coverage validity with only a logarithmic effect on the high-probability bound on the length of the confidence interval.  The only other price we pay is that the computational cost increases as $\beta$ decreases, as we describe below.

Our change to Algorithm~\ref{Alg:CI} is as follows: we replace the definition of $b_{\min}$ by setting
\[
b_{\min} = \frac{\beta}{\sqrt{2^{\lfloor \log_2 (2p)\rfloor} \log_2(2p)}} \wedge \frac{1}{\sqrt{2}},
\]
set $M = \lceil 2\log_2(1/b_{\min})\rceil$ and define $\mathcal{B} = \bigl\{\pm 2^{m/2} b_{\min}: m \in [M]\}$ in both the \texttt{ocd}$'$ base procedure and in Algorithm~\ref{Alg:CI}.  The rest of algorithm remains as previously stated.  Thus, if we choose a conservative (very small) $\beta$, then the effect of the modification is to increase the number of scales on which we search for a change, so that the largest element of $\mathcal{B}$ is of order 1.  In order to state our theoretical results for this modified algorithm, first define $b_{\mathrm{opt}} := \max \Bigl\{b \in \mathcal{B}\cap (0, \infty): b \leq \frac{\vartheta}{\sqrt{s\log_2(2p)}}\Bigr\}$, which satisfies 
$b_{\mathrm{opt}} \geq \frac{\vartheta}{\sqrt{2s\log_2(2p)}} \wedge 1$. Under the same assumptions as in Proposition~\ref{prop:ocdsatisfyassumption}, and modifying the inputs to $(X_t)_{t\in \mathbb{N}}$, $\beta > 0$, $\tilde{a} = \sqrt{2\log(16p^2\gamma M)}$, $T^{\mathrm{diag}} = \log\bigl(16p\gamma(M+1)\bigr)$ and $T^{\mathrm{off}} = 8\log(16p\gamma M)$, it can be shown using very similar arguments to those in the proof of Proposition~\ref{prop:ocdsatisfyassumption} that there exists a universal constant $C' > 0$ such that with
	$
	r \geq C'b_{\mathrm{opt}}^{-2}\log\bigl(p\gamma \alpha^{-1}(\beta^{-2}\vee 1)\bigr) =: r_1,
	$
the output $N$ satisfies
\[
    \mathbb{P}_{z,\theta,\Sigma}(N \leq z) + g(r; N) + 4rp^2(M+1)^2e^{-rb_{\mathrm{opt}}^2/8} \leq \frac{3}{4} \alpha.
\]
With this in place, we can derive Corollary~\ref{Cor:Unification} below, which is the analog of Corollary~\ref{cor:allresultsocd} for our modified algorithm.
\begin{cor}
	\label{Cor:Unification}
	Fix $\alpha \in (0, 1)$, $\gamma > 0$. Assume that $\theta$ has an effective sparsity of $s := s(\theta) \geq 2$, that $\vartheta \geq \beta > 0$ and that $z \leq 2\alpha \gamma$. Let $(X_t)_{t\in \mathbb{N}}$, $\beta > 0$, $\tilde{a} = \sqrt{2\log(16p^2\gamma M)}$, $T^{\mathrm{diag}} = \log\bigl(16p\gamma(M+1)\bigr)$ and $T^{\mathrm{off}} = 8\log(16p\gamma M)$ be the inputs of the modified $\mathtt{ocd}'$ procedure.  Let $d_1 = \sqrt{5C'\log\bigl(p\gamma \alpha^{-1}(\beta^{-2}\vee 1)\bigr)/9}$ and $d_2 = 4d_1^2$. Then the following statements hold:
	
	\medskip
	\noindent (a) With extra inputs $\mathrm{CP} = \mathtt{ocd}'$, $a \geq 0$, and $\ell \geq 0$ for the modified Algorithm~\ref{Alg:CI}, the output confidence interval $\mathcal{C}$ and the support estimate $\hat{\mathcal{S}}$ satisfy $\mathbb{P}_{z,\theta,\Sigma}( z \in \mathcal{C} ) \geq 1-\alpha$ and $\mathbb{P}_{z,\theta,\Sigma}(\hat{\mathcal{S}} \subseteq \mathcal{S}_{\beta}) \geq 1-\alpha$.
	
	\medskip
	\noindent (b) There exists a universal constant $C > 0$ such that, with extra inputs $\mathrm{CP} =$ modified $\mathtt{ocd}'$, $a = C\sqrt{\log\bigl(p\gamma \alpha^{-1}(\beta^{-2}\vee 1)\bigr)}$, and $\ell \geq 80r_1$ for the modified Algorithm~\ref{Alg:CI}, the length~$L$ of the output confidence interval $\mathcal{C}$ and the support estimate satisfy 
	\begin{equation}
	\label{Eq:ModifiedLength}
	\mathbb{P}_{z,\theta,\Sigma}\bigg(L > 8C' \max\biggl\{\frac{2s\log_2(2p)}{\vartheta^2}, 1\biggr\}\log\bigl(p\gamma \alpha^{-1}(\beta^{-2}\vee 1)\bigr)\biggr) \leq \alpha
	\end{equation}
	and $\mathbb{P}_{z,\theta,\Sigma}(\hat{\mathcal{S}}\cup \{\hat{j}\} \supseteq  \mathcal{S}) \geq 1-\alpha$.
	\end{cor}
The main difference between Corollaries~\ref{Cor:Unification} and \ref{cor:allresultsocd} concerns the high-probability guarantees on the length of the confidence interval.  Ignoring logarithmic factors, with high probability the length of the confidence interval in the modified algorithm is at most of order $(s/\vartheta^2) \vee 1$, whereas for the original algorithm it was of order $(s/\beta^2) \vee 1$.  Thus, the modified algorithm has the significant advantage of enabling a conservative choice of $\beta$ with only a logarithmic effect on the length guarantee relative to an oracle procedure with knowledge of $\|\theta\|_2$.  The computational complexity per new observation and the storage requirements of this modified algorithm are $O\Bigl(p^2\bigl(\log(ep) + \log(1/\beta)\bigr)\Bigr)$, so the order of magnitude is increased relative to the original algorithm only in an asymptotic regime where $\beta$ is small by comparison with $1/p^K$ for every $K > 0$.  Moreover, the modified algorithm still does not require storage of historical data and the computational time per new observation after observing $n$ observations does not increase with $n$.  Nevertheless, since the computational complexity now depends on $\beta$, the modified algorithm does not strictly satisfy our definition of an online algorithm given in introduction.

	\section{Numerical studies} \label{Sec:Simulation}
	
	In this section, we study the empirical performance of the \texttt{ocd\_CI} algorithm.  Throughout this section, by default, the \texttt{ocd\_CI} algorithm is used in conjunction with the recommended base online changepoint detection procedure $\mathrm{CP} = \texttt{ocd}$. 
		
	\subsection{Tuning Parameters}
	\label{Sec:Tuning}
	
	   \citet{CWS2021} found that the theoretical choices of thresholds $T^{\mathrm{diag}}$ and $T^{\mathrm{off}}$ for the \texttt{ocd} procedure were a little conservative, and therefore recommended determining these thresholds via Monte Carlo simulation; we replicate the method for choosing these thresholds described in their Section~4.1.  Likewise, as in \citet{CWS2021}, we take $a = \tilde{a} = \sqrt{2\log p }$ in our simulations.

	 For $d_1$ and $d_2$, as suggested by our theory, we take $d_2 = 4d_1^2$, and take $d_1$ to be of the form $d_1 = c\sqrt{\log(p/\alpha)}$.  Here, we tune the parameter $c > 0$ through Monte Carlo simulation, as we now describe.  We considered the parameter settings $p \in \{100, 500\}$, $s \in \{2, \lfloor \sqrt{p}\rfloor, p\}$, $\vartheta \in \{2,1,1/2\}$, $\Sigma = I_p$, $\alpha = 0.05$, $\beta \in \{2\vartheta, \vartheta, \vartheta/2\}$, $\gamma=30000$ and $z=500$.  Then, with $\theta$ generated as $\vartheta U$, where $U$ is uniformly distributed on the union of all $s$-sparse unit spheres in $\mathbb{R}^p$ (independent of our data), we studied the coverage probabilities, estimated over 2000 repetitions as $c$ varies, of the \texttt{ocd\_CI} confidence interval for data generated according to the Gaussian model defined at the beginning of Section~\ref{Sec:Theory}.  Figure~\ref{fig:choice_c} displays a subset of the results (the omitted curves were qualitatively similar).  On this basis, we recommend $c=0.5$ as a safe choice across a wide range of data generating mechanisms, and we used this value of $c$ throughout our confidence interval simulations.

	\begin{figure}[htbp!]
	\centering
	\includegraphics[width=\textwidth]{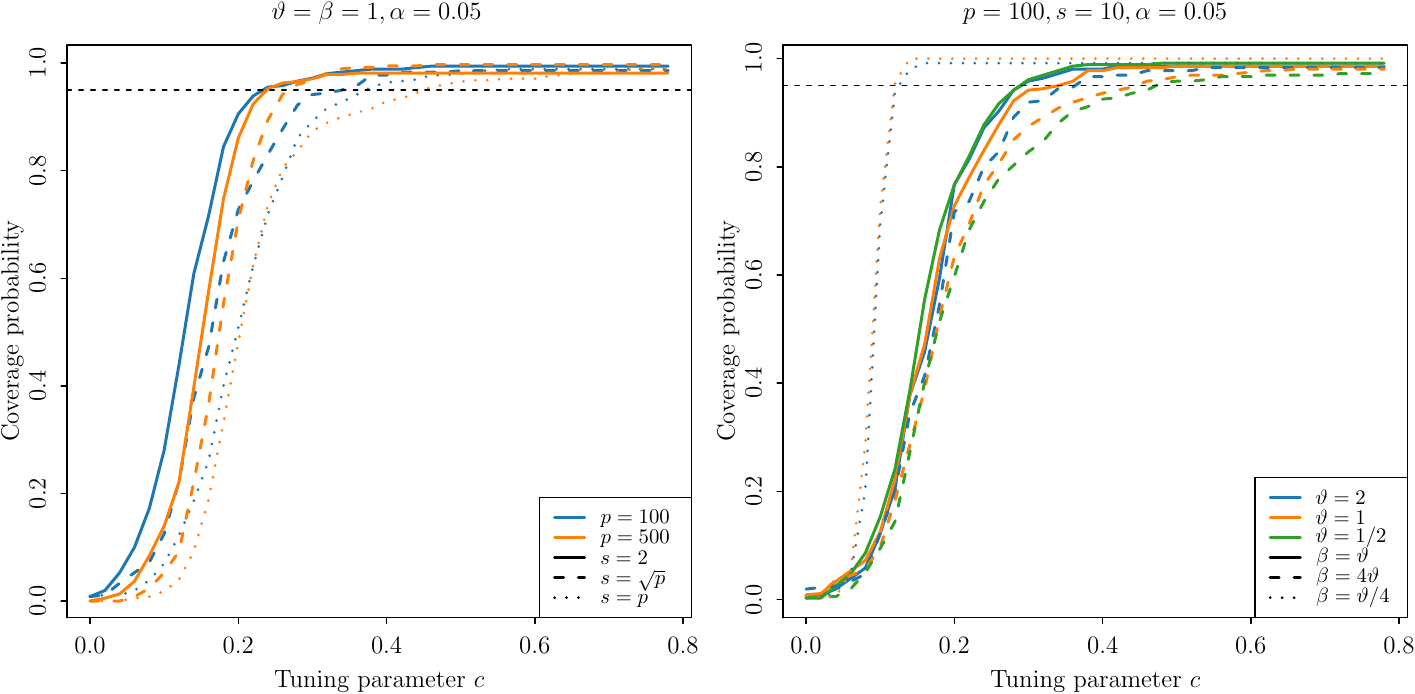}
	\caption{Coverage probabilities of the \texttt{ocd\_CI} confidence interval as the parameter $c$, involved in the choice of tuning parameter $d_1$, varies.}
	\label{fig:choice_c}
	\end{figure}

	The previous two paragraphs, in combination with Algorithms~\ref{Alg:CI} and~\ref{Alg:ocd}, provide the practical implementation of the \texttt{ocd\_CI} algorithm that we use in our numerical studies and that we recommend for practitioners.  The only quantity that remains for the practitioner to input (other than the data) is $\beta$, which represents a lower bound on the Euclidean norm of the vector of mean change.  Fortunately, this description makes $\beta$ easily interpretable by practitioners.  In cases where an informed default choice is not available, practitioners can make a conservative (very small) choice and use an increased grid of scales to with only a small inflation in the confidence interval length guarantee and computational cost; see Section~\ref{SubSec:Beta}. 
	
	\subsection{Coverage Probability and Interval Length}
	\label{Sec:CoverageLength}
	
	In Table~\ref{Tab:CI_cov_len}, we present the detection delay of the \texttt{ocd} procedure, as well as the coverage probabilities and average confidence interval lengths of the \texttt{ocd\_CI} procedure, all estimated over 2000 repetitions, with the same set of parameter choices and data generating mechanism as in Section~\ref{Sec:Tuning}.  From this table, we see that the coverage probabilities are at least at the nominal level (up to Monte Carlo error) across all settings considered.  Underspecification of~$\beta$ means that the grid of scales that can be chosen for indices in $\hat{\mathcal{S}}$ is shifted downwards, and therefore increases the probability that $\tilde{b}^j$ will underestimate $\theta^j$ for $j \in \hat{\mathcal{S}}$.  In turn, this leads to a slight conservativeness for the coverage probability (and corresponding increased average confidence interval length).  On the other hand, overspecification of $\beta$ yields a shorter interval on average, though these were nevertheless able to retain the nominal coverage in all cases considered.

	\begin{table}[htbp!]
		\footnotesize
        \begin{center}
        	 \caption{\label{Tab:CI_cov_len} Estimated coverage and average length of the \texttt{ocd\_CI} confidence interval and average detection delay over 2000 repetitions, with standard errors in brackets.  Other parameters: $\gamma=30000$, $z=1000$, $\Sigma = I_p$, $\alpha = 0.05$, $a = \tilde{a} = \sqrt{2 \log p}$, $c=0.5$, $d_1 = c\sqrt{\log(p/\alpha)}$, $d_2 = 4d_1^2$. For comparison, we also present the corresponding estimated coverage probabilities and average lengths of the procedure based on \citet{kaul2021inference}, as described in Section~\ref{Sec:CoverageLength}.}
	        \begin{tabular}{cccc|ccc|cc}
	            \hline\hline
	              &  &  &  & \multicolumn{3}{c|}{\texttt{ocd\_CI}} & \multicolumn{2}{c}{\texttt{Kaul\_et\_al}} \\
                $p$ & $s$ & $\vartheta$ & $\beta$ & Delay & Coverage (\%) & CI Length & Coverage (\%) & CI Length\\
                \hline
                $100$ & $2$ & $2$ & $4$ & $9.8_{(0.1)}$ & $96.2_{(0.4)}$ & $20.1_{(0.7)}$ & $83.2_{(0.8)}$ & $732.7_{(9.6)}$ \\
                $100$ & $2$ & $2$ & $2$ & $12.6_{(0.1)}$ & $97.0_{(0.4)}$ & $33.7_{(0.7)}$ & $82.5_{(0.8)}$ & $474.9_{(11.0)}$ \\
                $100$ & $2$ & $2$ & $1$ & $14.1_{(0.1)}$ & $97.9_{(0.3)}$ & $80.8_{(1.0)}$ & $83.5_{(0.8)}$ & $341.4_{(10.4)}$ \\
                $100$ & $2$ & $1$ & $2$ & $34.2_{(0.3)}$ & $95.8_{(0.4)}$ & $66.1_{(1.0)}$ & $76.6_{(0.9)}$ & $399.3_{(10.5)}$ \\
                $100$ & $2$ & $1$ & $1$ & $44.2_{(0.3)}$ & $97.5_{(0.4)}$ & $122.0_{(1.4)}$ & $80.5_{(0.9)}$ & $123.8_{(6.5)}$ \\
                $100$ & $2$ & $1$ & $0.5$ & $52.0_{(0.4)}$ & $97.4_{(0.4)}$ & $309.1_{(2.0)}$ & $81.5_{(0.9)}$ & $90.8_{(5.4)}$  \\ \hline
                $100$ & $10$ & $2$ & $4$ & $14.7_{(0.1)}$ & $96.0_{(0.4)}$ & $32.5_{(0.8)}$ & $80.2_{(0.9)}$ & $636.4_{(10.5)}$  \\
                $100$ & $10$ & $2$ & $2$ & $15.7_{(0.1)}$ & $97.4_{(0.4)}$ & $38.4_{(0.8)}$ & $77.4_{(0.9)}$ & $537.9_{(10.9)}$ \\
                $100$ & $10$ & $2$ & $1$ & $15.9_{(0.1)}$ & $97.0_{(0.4)}$ & $80.2_{(1.1)}$ & $80.8_{(0.9)}$ & $542.4_{(10.9)}$ \\
                $100$ & $10$ & $1$ & $2$ & $52.6_{(0.5)}$ & $96.2_{(0.4)}$ & $114.0_{(1.5)}$ & $75.8_{(1)}$ & $342.4_{(10.1)}$ \\
                $100$ & $10$ & $1$ & $1$ & $56.9_{(0.4)}$ & $97.1_{(0.4)}$ & $142.5_{(1.8)}$ & $73.9_{(1)}$ & $262.6_{(9.1)}$  \\
                $100$ & $10$ & $1$ & $0.5$ & $60.2_{(0.4)}$ & $98.2_{(0.3)}$ & $301.1_{(1.6)}$ & $75.9_{(1)}$ & $248.3_{(8.9)}$  \\ \hline
                $100$ & $100$ & $2$ & $4$ & $27.2_{(0.2)}$ & $96.1_{(0.4)}$ & $77.6_{(0.9)}$ & $68.2_{(1.0)}$ & $533.9_{(10.7)}$ \\
                $100$ & $100$ & $2$ & $2$ & $27.7_{(0.2)}$ & $96.0_{(0.4)}$ & $81.8_{(1.0)}$ & $71.3_{(1.0)}$ & $537.7_{(10.8)}$ \\
                $100$ & $100$ & $2$ & $1$ & $28.2_{(0.2)}$ & $97.5_{(0.3)}$ & $99.4_{(1.3)}$ & $71.8_{(1.0)}$ & $556.0_{(10.7)}$ \\
                $100$ & $100$ & $1$ & $2$ & $100.7_{(0.8)}$ & $94.7_{(0.5)}$ & $292.8_{(3.5)}$ & $87.7_{(0.7)}$ & $850.5_{(9.5)}$ \\
                $100$ & $100$ & $1$ & $1$ & $100.5_{(0.9)}$ & $96.3_{(0.4)}$ & $296.0_{(3.4)}$ & $88.0_{(0.7)}$ & $863.7_{(9.3)}$ \\
                $100$ & $100$ & $1$ & $0.5$ & $103.2_{(0.8)}$ & $97.3_{(0.4)}$ & $365.9_{(2.8)}$ & $89.3_{(0.7)}$ & $876.8_{(9.1)}$ \\ \hline
                $500$ & $2$ & $2$ & $4$ & $11.3_{(0.1)}$ & $97.2_{(0.4)}$ & $23.1_{(0.7)}$ & $92.0_{(0.6)}$ & $958.7_{(4.3)}$ \\
                $500$ & $2$ & $2$ & $2$ & $15.8_{(0.1)}$ & $97.7_{(0.3)}$ & $45.2_{(0.9)}$ & $83.3_{(0.8)}$ & $806.4_{(8.7)}$ \\
                $500$ & $2$ & $2$ & $1$ & $17.7_{(0.1)}$ & $97.5_{(0.4)}$ & $117.3_{(1.0)}$ & $79.9_{(0.9)}$ & $624.9_{(10.7)}$ \\
                $500$ & $2$ & $1$ & $2$ & $41.5_{(0.3)}$ & $97.3_{(0.4)}$ & $81.8_{(1.2)}$ & $80.0_{(0.9)}$ & $774.9_{(9.4)}$ \\
                $500$ & $2$ & $1$ & $1$ & $55.0_{(0.4)}$ & $96.8_{(0.4)}$ & $168.9_{(1.5)}$ & $73.0_{(1)}$ & $275.9_{(9.4)}$ \\
                $500$ & $2$ & $1$ & $0.5$ & $64.6_{(0.5)}$ & $98.1_{(0.3)}$ & $445.0_{(1.7)}$ & $75.6_{(1)}$ & $186.1_{(8.0)}$  \\ \hline
                $500$ & $22$ & $2$ & $4$ & $23.6_{(0.2)}$ & $96.3_{(0.4)}$ & $55.4_{(1.0)}$ & $87.0_{(0.8)}$ & $884.9_{(7.3)}$ \\
                $500$ & $22$ & $2$ & $2$ & $25.0_{(0.2)}$ & $97.0_{(0.4)}$ & $60.3_{(0.8)}$ & $85.5_{(0.8)}$ & $864.2_{(7.8)}$ \\
                $500$ & $22$ & $2$ & $1$ & $25.5_{(0.2)}$ & $98.1_{(0.3)}$ & $119.7_{(0.8)}$ & $83.6_{(0.8)}$ & $823.0_{(8.6)}$ \\
                $500$ & $22$ & $1$ & $2$ & $88.1_{(0.7)}$ & $97.0_{(0.4)}$ & $203.5_{(2.1)}$ & $77.2_{(0.9)}$ & $645.0_{(11.0)}$ \\
                $500$ & $22$ & $1$ & $1$ & $91.9_{(0.6)}$ & $97.8_{(0.3)}$ & $229.7_{(2.2)}$ & $76.2_{(1)}$ & $562.8_{(11.1)}$ \\
                $500$ & $22$ & $1$ & $0.5$ & $94.9_{(0.6)}$ & $98.3_{(0.3)}$ & $462.8_{(1.4)}$ & $75.5_{(1)}$ & $538.3_{(11.2)}$  \\ \hline
                $500$ & $500$ & $2$ & $4$ & $79.8_{(0.6)}$ & $95.0_{(0.5)}$ & $238.9_{(2.7)}$ & $88.5_{(0.7)}$ & $913.0_{(8.0)}$ \\
                $500$ & $500$ & $2$ & $2$ & $80.3_{(0.6)}$ & $95.8_{(0.4)}$ & $245.7_{(2.6)}$ & $90.3_{(0.7)}$ & $928.8_{(7.7)}$ \\
                $500$ & $500$ & $2$ & $1$ & $80.9_{(0.6)}$ & $97.5_{(0.4)}$ & $250.2_{(2.5)}$ & $90.6_{(0.7)}$ & $928.3_{(7.7)}$ \\
                $500$ & $500$ & $1$ & $2$ & $290.5_{(2.3)}$ & $94.5_{(0.5)}$ & $819.7_{(7.9)}$ & $95.2_{(0.5)}$ & $1189.4_{(7.3)}$ \\
                $500$ & $500$ & $1$ & $1$ & $291.4_{(2.3)}$ & $95.2_{(0.5)}$ & $831.1_{(7.5)}$ & $94.3_{(0.5)}$ & $1204.9_{(7.0)}$ \\
                $500$ & $500$ & $1$ & $0.5$ & $297.3_{(2.3)}$ & $98.1_{(0.3)}$ & $875.0_{(6.7)}$ & $94.6_{(0.5)}$ & $1207.4_{(6.8)}$ \\
                \hline
            \end{tabular}
        \end{center}
    \end{table}

	Another interesting feature of Table~\ref{Tab:CI_cov_len} is to compare the average confidence interval lengths with the corresponding average detection delays.  Corollary~\ref{cor:allresultsocd}(b), as well as \citet[][Theorem~4]{CWS2021}, indicates that both of these quantities are of order $(s/\beta^2) \vee 1$, up to polylogarithmic factors in $p$ and $\gamma$, but of course whenever the confidence interval includes the changepoint, its length must be at least as long as the detection delay.  Nevertheless, in most settings, it is only 2 to 3 times longer on average, and in all cases considered was less than 7 times longer on average.  Moreover, we can also observe that the confidence interval length increases with $s$ and decreases with $\beta$, as anticipated by our theory. 
	
	For comparison, we also present the corresponding coverage probabilities and average lengths of confidence intervals obtained using an offline procedure as described in the introduction. More precisely, after the \texttt{ocd} algorithm has declared a change, we treat the data up to the stopping time as an offline dataset, and apply the \texttt{inspect} algorithm \citep{WangSamworth2018}, followed by the one-step refinement of \citet{kaul2021inference}, to construct an estimate, $\hat{z}^{\texttt{KFJS}}$, of the changepoint location. As recommended by \citet{kaul2021inference}, we obtain an estimator $\hat\vartheta^{\texttt{KFJS}}$ of $\vartheta$ using the $\ell_2$-norm of the soft-thresholded difference in empirical mean vectors before and after $\hat{z}^{\texttt{KFJS}}$, with the soft-thresholding parameter chosen via the Bayesian Information Criterion. The final confidence interval is then of the form $[\hat{z}^{\texttt{KFJS}} - q^{\alpha/2}/(\hat{\vartheta}^{\texttt{KFJS}})^2, \hat{z}^{\texttt{KFJS}} + q^{\alpha/2}/(\hat{\vartheta}^{\texttt{KFJS}})^2]$, where $q^{\alpha/2}$ is the $1-\alpha/2$ quantile of the distribution of the (almost surely unique) maximizer of a two-sided Brownian motion with a triangular drift as given by \citet[Theorem~3.1]{kaul2021inference}. In particular, we have $q^{0.025} = 11.03$. The last two columns of Table~\ref{Tab:CI_cov_len} reveal that both the coverage probabilities and confidence interval lengths from this procedure are disappointing and not competitive with those of the \texttt{ocd\_CI} algorithm. There are two main reasons for this: first, the nature of the online problem means that the changepoint is often located near the right-hand end of the dataset up to the stopping time; on the other hand, the theoretical guarantees of \citet{kaul2021inference} are obtained under an asymptotic setting where the fraction of data either side of the change is bounded away from zero. Thus, the estimated changepoint from the one-step refinement is often quite poor. Moreover, the estimated magnitude of change, $\hat{\vartheta}^{\texttt{KFJS}}$, is often a significant underestimate of $\vartheta$ due to the soft-thresholding operation, and this can lead to substantially inflated confidence interval lengths. 
	We emphasize that the \citet{kaul2021inference} procedure was not designed for use in this online setting, but we nevertheless present these results to illustrate the fact that the naive application of offline methods in sequential problems may fail badly.
	
	While Table~\ref{Tab:CI_cov_len} covers the most basic setting for our methodology, our theory in Section~\ref{Sec:Theory} applies equally well to data with 
	spatial dependence across different coordinates.  To assess whether this theory carries over to empirical performance, Table~\ref{Tab:Spatial} in the supplement (Section~\ref{Sec:AddSims}) presents corresponding coverage probabilities and lengths for the \texttt{ocd\_CI} procedure with the cross-sectional covariance matrix $\Sigma = (\Sigma_{jk})_{j,k \in [p]}$ taken to be Toeplitz with parameter $\rho \in \{0.5,0.75\}$; in other words, $\Sigma_{jk} = \rho^{|j-k|}$.  The results are again encouraging: the coverage remains perfectly satisfactory in all settings considered, and moreover, the lengths of the confidence intervals are very similar to those in Table~\ref{Tab:CI_cov_len}.

\begin{table}[htbp!]
	\footnotesize
        \begin{center}
        	\caption{\label{Tab:supp_rec_with_extra} Estimated support recovery probabilities (with standard errors in brackets).  Other parameters: $p=100$, $\Sigma = I_p$, $\alpha = 0.05$, $a = \tilde{a} = \sqrt{2 \log p}$, $d_1 = \sqrt{2 \log (p/\alpha)}$, $\beta=\vartheta$, and with an additional $\ell = \lceil 2s\log_2(2p)\log(p)\beta^{-2}\rceil$ post-declaration observations.}
            \begin{tabular}{ccccc}
                \hline\hline
                $s$ & $\vartheta$ & Signal Shape & $\hat{\mathcal{S}} \subseteq \mathcal{S}_{\beta}$ (\%) & $\hat{\mathcal{S}}\cup{\{\hat{j}\}} \supseteq \mathcal{S}$ (\%)  \\
                \hline
                $5$ & $2$ & uniform & $99.8_{(0.2)}$ & $97.6_{(0.7)}$\\
                $5$ & $1$ & uniform & $100.0_{(0.0)}$ & $97.6_{(0.7)}$\\
                $50$ & $2$ & uniform & $100.0_{(0.0)}$ & $95.6_{(0.9)}$\\
                $50$ & $1$ & uniform & $100.0_{(0.0)}$ & $97.8_{(0.7)}$ \vspace{1.5mm}\\
                $5$ & $2$ & inv sqrt & $99.6_{(0.3)}$ & $96.6_{(0.8)}$\\
                $5$ & $1$ & inv sqrt & $100.0_{(0.0)}$ & $98.8_{(0.5)}$\\
                $50$ & $2$ & inv sqrt & $100.0_{(0.0)}$ & $99.8_{(0.2)}$\\
                $50$ & $1$ & inv sqrt & $100.0_{(0.0)}$ & $100.0_{(0.0)}$ \vspace{1.5mm} \\
                $5$ & $2$ & harmonic & $100.0_{(0.0)}$ & $97.6_{(0.7)}$\\
                $5$ & $1$ & harmonic & $99.6_{(0.3)}$ & $97.8_{(0.7)}$\\
                $50$ & $2$ & harmonic & $100.0_{(0.0)}$ & $99.4_{(0.3)}$\\
                $50$ & $1$ & harmonic & $100.0_{(0.0)}$ & $100.0_{(0.0)}$\\
                \hline
            \end{tabular}
        \end{center}
    \end{table}
	

	\subsection{Support Recovery}
	\label{Sec:Simualtion_support}
	We now turn our attention to the empirical support recovery properties of the quantity $\hat{\mathcal{S}}$ (in combination with the anchor coordinate $\hat{j}$) computed in the \texttt{ocd\_CI} algorithm.  In Table~\ref{Tab:supp_rec_with_extra}, we present the probabilities, estimated over 500 repetitions, that $\hat{\mathcal{S}} \subseteq \mathcal{S}_\beta$ and that $\hat{\mathcal{S}} \cup \{\hat{j\}} \supseteq \mathcal{S}$ for $p =100$, $s \in \{5,50\}$, $\vartheta \in \{1,2\}$, $\Sigma = I_p$, and for three different signal shapes: in the uniform, inverse square root and harmonic cases, we took $\theta \propto (\mathbbm{1}_{\{j \in [s]\}})_{j \in [p]}$, $\theta \propto (j^{-1/2}\mathbbm{1}_{\{j \in [s]\}})_{j \in [p]}$ and $\theta \propto (j^{-1}\mathbbm{1}_{\{j \in [s]\}})_{j \in [p]}$ respectively.  As inputs to the algorithm, we set $a = \tilde{a} = \sqrt{2 \log p}$, $\alpha = 0.05$, $d_1 = \sqrt{2 \log (p/\alpha)}$, $\beta = \vartheta$, and, motivated by Corollary~\ref{cor:allresultsocd}, took an additional $\ell = \lceil 2s\beta^{-2}\log_2(2p)\log p\rceil$ post-declaration observations in constructing the support estimates.  The results reported in Table~\ref{Tab:supp_rec_with_extra} provide empirical confirmation of the support recovery properties claimed in Corollary~\ref{cor:allresultsocd}.

Finally in this section, we consider the extent to which the  additional $\ell$ observations are necessary in practice to provide satisfactory support recovery.  In the left panel of Figure~\ref{fig:support}, we plot Receiver Operating Characteristic (ROC) curves to study the estimated support recovery probabilities with $\ell = 0$ as a function of the input parameter $d_1$, which can be thought of as controlling the trade-off between $\mathbb{P}(\hat{\mathcal{S}} \cup \{\hat{j}\} \supseteq \mathcal{S})$ and $\mathbb{P}(\hat{\mathcal{S}} \subseteq \mathcal{S}_\beta)$.  The fact that the triangles in this plot are all to the left of the dotted vertical line confirms the theoretical guarantee provided in Corollary~\ref{cor:allresultsocd}(a), which holds with $d_1 = \sqrt{2 \log (p/\alpha)}$, and even with $\ell=0)$; the less conservative choice $d_1 = \sqrt{2 \log p}$, which roughly corresponds to an average of one noise coordinate included in $\hat{\mathcal{S}}$, allows us to capture a larger proportion of the signal.  From this panel, we also see that additional sampling is needed to ensure that, with high probability, we recover all of the true signals.  This is unsurprising: for instance, with a uniform signal shape and $s=50$, it is very unlikely that all 50 signal coordinates will have accumulated such similar levels of evidence to appear in $\hat{\mathcal{S}} \cup \{\hat{j}\}$ by the time of declaration.  The right panel confirms that, with an inverse square root signal shape, the probability that we capture each signal increases with the signal magnitude, and that even small signals tend to be selected with higher probability than noise coordinates.

    \begin{figure}[htbp!]
	\centering
	\includegraphics[width=0.85\textwidth]{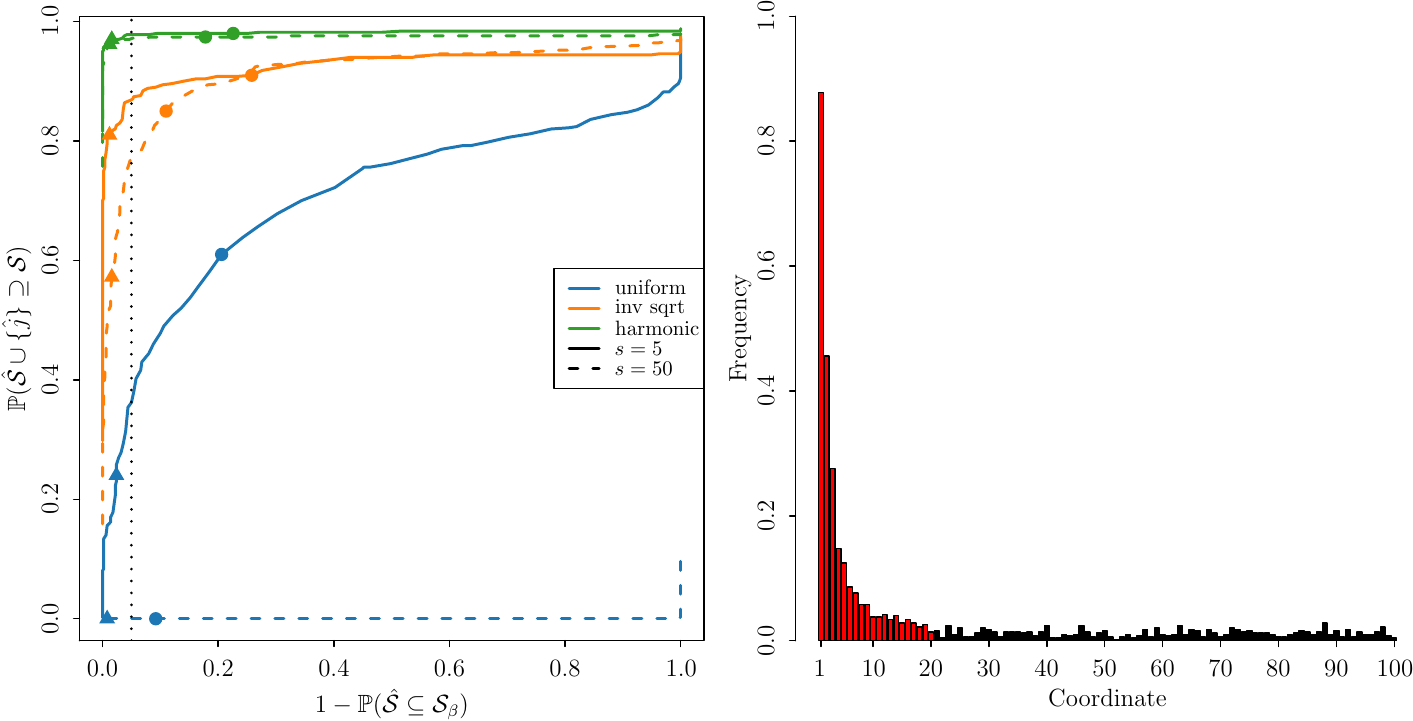}
	\caption{Support recovery properties of \texttt{ocd\_CI}.  In the left panel, we plot ROC curves for three different signal shapes and for sparsity levels $s \in \{5,50\}$.  The triangles and circles correspond to points on the curves with $d_1 = \sqrt{2 \log (p/\alpha)}$ (with $\alpha= 0.05$), and $d_1 = \sqrt{2 \log p}$ respectively.  The dotted vertical line corresponds to $\mathbb{P}(\hat{\mathcal{S}} \subseteq \mathcal{S}_\beta) = 1 - \alpha$.  In the right panel, we plot the proportion of 500 repetitions for which each coordinate belongs to $\hat{\mathcal{S}} \cup \{\hat{j}\}$ with $d_1 = \sqrt{2 \log p}$; here, the $s=20$ signals have an inverse square root shape, and are plotted in red; noise coordinates are plotted in black.  Other parameters for both panels: $p=100$, $\Sigma = I_p$, $\beta = \vartheta = 2$, $\ell=0$, $a = \tilde{a} = \sqrt{2 \log p}$.}
	\label{fig:support}
	\end{figure}

\subsection{US Covid-19 Data Example}
\label{Sec:USCovid}

In this section, we apply \texttt{ocd\_CI} to a dataset of weekly deaths in the United States between January 2017 and June 2020 (available at: \url{https://www.cdc.gov/nchs/nvss/vsrr/covid19/excess_deaths.htm}).  The data up to 29 June 2019 are treated as our training data.  An obvious discrepancy between underlying dynamics of these weekly deaths and the conditions assumed in our theory in Section~\ref{Sec:Theory} is temporal dependence, particularly induced by seasonal and weather effects.  Although we can never hope to remove this dependence entirely, we seek to mitigate its impact by pre-processing the data as follows: for each of the 50 states, as well as Washington, D.C. ($p = 51$), we first estimate the ``seasonal death curve'', i.e.~the mean death numbers for each day of the year, for each state. The seasonal death curve is estimated by first splitting the weekly death numbers evenly across the seven relevant days, and then estimating the average number of deaths on each day of the year from these derived daily death numbers using a Gaussian kernel with a bandwidth of 20 days. As the death numbers follow an approximate Poisson distribution, we apply a square-root transformation to stabilize the variance; more precisely, the transformed weekly excess deaths are computed as the difference of the square roots of the weekly deaths and the predicted weekly deaths from the seasonal death curve. Finally, we standardize the transformed weekly excess deaths using the mean and standard deviation of the transformed data over the training period. The standardized, transformed  data are plotted in Figure~\ref{fig:US_Covid_data} for 12  states.

\begin{figure}[htbp!]
	\centering
	\includegraphics[width=\textwidth]{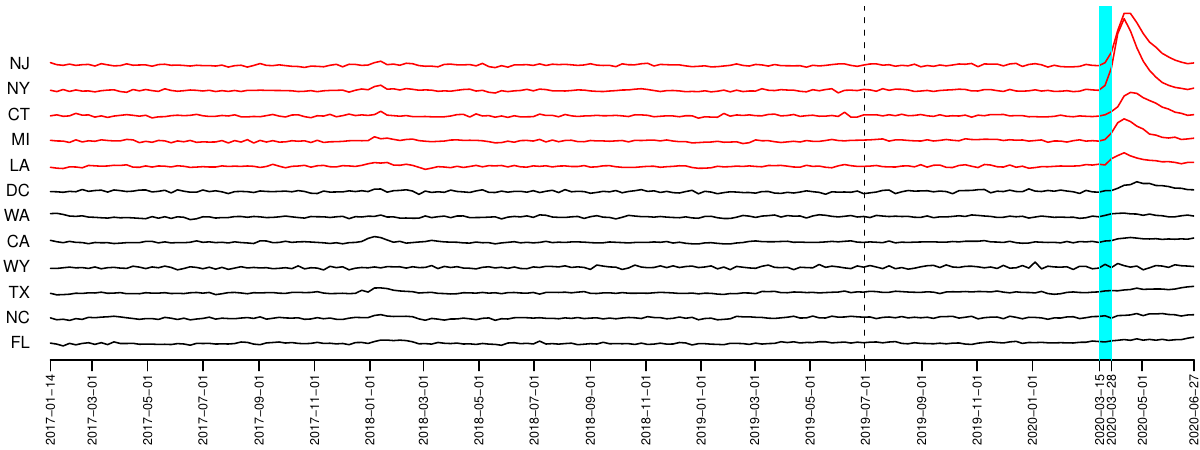}
	\caption{Standardized, transformed weekly excess death data from 12 states (including Washington, D.C.). The monitoring period starts from 30 June 2019 (dashed line). The data from the states in the support estimate are shown in red. The confidence interval [8 March 2020, 28 March 2020] is shown in the light blue shaded region.}
	\label{fig:US_Covid_data}
\end{figure}

When applying \texttt{ocd\_CI} to these data, we take $a = \tilde{a} = \sqrt{2\log p}$, $T^{\mathrm{diag}} = \log\{16p\gamma\log_2(4p)\}$, $T^{\mathrm{off}}=8\log\{16p\gamma\log_2(2p)\}$, $d_1 = 0.5\sqrt{\log (p/\alpha)}$ and $d_2 = 4d_1^2$, with $\alpha = 0.05$, $\beta = 50$ and $\gamma = 1000$. On the monitoring data (from 30 June 2019), the \texttt{ocd\_CI} algorithm declares a change on the week ending 28 March 2020, and provides a confidence interval from the week ending 21 March 2020 to the week ending 28 March 2020. This coincides with the beginning of the first wave of Covid-19 deaths in the United States. The algorithm also identifies New York, New Jersey, Connecticut, Michigan and Louisiana as the estimated support of the change.  Interestingly, if we run the \texttt{ocd\_CI} procedure from the beginning of the training data period  (while still standardizing as before, due to the lack of available data prior to 2017), it identifies a subtler change on the week ending 6 January 2018, with a confidence interval of [17 December 2017, 6 January 2018].  This corresponds to a bad influenza season at the end of 2017 (see, \url{https://www.cdc.gov/flu/about/season/flu-season-2017-2018.htm}).  Despite the natural interpretation of these findings, we recognize that the model in Section~\ref{Sec:Theory} under which we proved our theoretical results cannot capture the full complexity of the temporal dependence in this dataset even after our pre-processing transformations.  A complete theoretical analysis of the performance of \texttt{ocd\_CI} in time-dependent settings is challenging and beyond the scope of the current work; in practical applications, we advise careful modeling of this dependence to facilitate the construction of appropriate residuals for which the main effects of this dependence have been removed.

\section{Supplementary material}
\label{Sec:Supp}
In this section, we provide proofs of our main results (Section~\ref{Sec:Proofs}), auxiliary results and their proofs (Section~\ref{Sec:Auxiliary}), pseudocode for the \texttt{ocd} and \texttt{ocd}$'$ base procedures (Section~\ref{Sec:ocdprimealgorithm}), extensions of our results to sub-Gaussian and sub-exponential settings (Section~\ref{Sec:BeyondGaussianity}) and additional simulation results (Section~\ref{Sec:AddSims}). Throughout this section, we use $\mathbb{P}$ instead of $\mathbb{P}_{z,\theta,\Sigma}$ when it is clear from the context.

	\subsection{Proofs of main results}
	\label{Sec:Proofs}
	\begin{proof}[Proof of Theorem~\ref{Thm:Coverage}]
		Fix $r \geq 1$ that satisfies the assumption~\eqref{Eq:detectassumption} in the theorem, $n>z, j\in [p]$, $b \in \mathcal{B}$ and $j' \in [p]\setminus \{j\}$. We assume, without loss of generality, that $\theta^{j'} \geq 0$. The case $\theta^{j'} < 0$ can be analyzed similarly. Recall that $b_{\min}$, defined in Algorithm~\ref{Alg:CI}, is the smallest positive scale in $\mathcal{B}\cup\mathcal{B}_0$, and write $b_{\mathrm{aux}}^{j'} := \min\bigl\{ b \in (\mathcal{B}\cup\mathcal{B}_0) \cap (0, \infty): b \geq \theta^{j'}\bigr\}$.  Then we have $A_{n,b}^{j', j} + \sum_{i=n+1}^{n+\ell} X_i^{j'} \mid t_{n, b}^{j} \sim \mathcal{N}\bigl( \theta^{j'} \min\{n+\ell-z, t_{n, b}^{j}+\ell\},\, t_{n, b}^{j} +\ell \bigr)$. Thus, recalling the definition of $\hat{\mathcal{S}}$ and $\tilde{b}^{j'}$ from Algorithm~\ref{Alg:CI}, we have 
		\begin{align}
			\mathbb{P}\bigl( \{j' \in \hat{\mathcal{S}}\} &\cap \{\tilde b^{j'}\notin (0, \theta^{j'})\} \cap \{N=n, \hat{j} = j, \hat{b} = b\}\bigr) \nonumber \\
			&= \mathbb{E}\Bigl\{\mathbb{P}\Bigl( \{j' \in \hat{\mathcal{S}} \} \cap \{  \tilde{b}^{j'} \notin (-b_{\min}, b^{j'}_{\mathrm{aux}}) \} \cap \{ N=n, \hat{j} = j, \hat{b} = b\}\Bigm| t_{n,b}^{j}\Bigr)\Bigr\} \nonumber\\	
			&\leq  \mathbb{E} \biggl\{ \mathbb{P}\biggl(A_{n,b}^{j', j} +\sum_{i=n+1}^{n+\ell} X_i^{j'} \geq b_{\mathrm{aux}}^{j'} (t_{n,b}^{j}+\ell)+ d_1\bigl(t_{n,b}^{j}+\ell)^{1/2} \biggm| t_{n,b}^{j}\biggr)\biggr\}\nonumber\\
			&\hspace{1.8cm}  +  \mathbb{E} \biggl\{ \mathbb{P}\biggl(A_{n,b}^{j', j} +\sum_{i=n+1}^{n+\ell} X_i^{j'} \leq -b_{\min} (t_{n,b}^{j}+\ell) - d_1\bigl(t_{n,b}^{j}+\ell\bigr)^{1/2} \biggm| t_{n,b}^{j}\biggr)\biggr\} \nonumber\\
			& \leq  \mathbb{E}\bigl\{\bar\Phi\bigl((b^{j'}_{\mathrm{aux}} - \theta^{j'})(t_{n,b}^{j}+\ell)^{1/2} + d_1\bigr)\bigr\} + \mathbb{E}\bigl\{\bar\Phi\bigl((b_{\mathrm{min}} + \theta^{j'})(t_{n,b}^{j}+\ell)^{1/2} + d_1\bigr) \bigr\} 
			\nonumber\\
			&\leq 2\bar{\Phi}(d_1). \label{Eq:CoveragePfEq1} 
		\end{align}
		Moreover, by a similar argument to \eqref{Eq:1DMain} in the proof of Proposition~\ref{Prop:1DCI}, for $b\in (0,\theta^{j'})$, we have
		\begin{align} 
			&\mathbb{P}\bigl( n-t_{n,b}^{j'}-d_2/b^2 > z \bigr) \leq   2\bar{\Phi}\biggl(  \frac{\sqrt{d_2}}{b}(\theta^{j'}-b/2) \biggr)\leq 2\bar{\Phi}\bigl(\sqrt{d_2}/2\bigr). \label{Eq:CoveragePfEq2} 
		\end{align}
		Combining~\eqref{Eq:CoveragePfEq1} and~\eqref{Eq:CoveragePfEq2}, we have
		\begin{align*}
			\mathbb{P}\Bigl(\{j' &\in \hat{\mathcal{S}}\}\cap \{n- t_{n, \tilde{b}^{j'}}^{j'}- d_2/(\tilde{b}^{j'})^2 > z\} \cap \{N=n, \hat{j}=j,\hat{b}=b\}\Bigr) \\
			&\leq  \mathbb{P}\Bigl( \{j' \in \hat{\mathcal{S}}\} \cap \{\tilde b^{j'} \notin (0, \theta^{j'})\} \cap \{N=n, \hat{j}=j,\hat{b}=b\}\Bigr) + \sum_{ b \in (\mathcal{B}\cup\mathcal{B}_0)\cap(0, \theta^j)} 2\bar{\Phi}\bigl(\sqrt{d_2}/2\bigr) \\
			&\leq 2\bar{\Phi}(d_1) +  2\log_2(4p)\bar{\Phi}\bigl(\sqrt{d_2}/2\bigr)  \leq 2\log_2(4p) e^{-5r\beta^2/(18s\log_2(2p))}, 
		\end{align*}
		where the last inequality follows from the choice of $d_1$ and $d_2$ in the statement of the theorem and the standard Gaussian tail bound used at the end of the proof of Lemma~\ref{Lemma:1DLemma}. By a union bound and~\eqref{Eq:detectassumption}, we have
		\begin{align*}
		&\mathbb{P}(z \notin \mathcal{C}) \leq \mathbb{P}(N \leq z) + \mathbb{P}\biggl( N-\min_{j \in \hat{\mathcal{S}}} \Bigl\{t_{N, \tilde{b}^j}^j +  \frac{d_2}{(\tilde{b}^j)^2}  \Bigr\}> z\biggr) \\
			& \leq \mathbb{P}(N \leq z) + \mathbb{P}(N > z+r) \\
			&\hspace{1cm}+ \sum_{n=z+1}^{z+\lfloor r\rfloor}  \sum_{j=1}^p \sum_{b \in\mathcal{B}}\sum_{j'=1}^p   \mathbb{P} \biggl(  \{j' \in \hat{\mathcal{S}}\} \cap \Bigl\{n- t_{n, \tilde{b}^{j'}}^{j'} - \frac{d_2}{(\tilde{b}^{j'})^2}  > z\Bigr\}  \cap \{N =n, \hat{j}=j, \hat{b}=b\} \biggr) \\
			&\leq \mathbb{P}(N \leq z) +  g(r;N) +  4rp^2\log_2^2(4p) e^{-5r\beta^2/(18s\log_2(2p))} \leq \alpha, 
		\end{align*}
		as required.
	\end{proof}
	
	\begin{proof}[Proof of Theorem~\ref{Thm:Length}]
		Fix $r \geq 1$ that satisfies the assumption~\eqref{Eq:detectassumption}. Denote $\ell_0 := 80r$. Then $\ell \geq \ell_0$. Since the output of Algorithm~\ref{Alg:CI} remains unchanged if we replace $(X_t^j:t\in\mathbb{N})$ by $(-X_t^j:t\in\mathbb{N})$ for any fixed $j$, we may assume without loss of generality that $\theta^1 \geq \theta^2 \geq \vartheta/\sqrt{s\log_2(2p)}$. For $j \in \{1, 2\}$, we denote $b^j := \max\{b \in \mathcal{B}\cup\mathcal{B}_0: b \leq \theta^j\}$ . Since $\vartheta \geq \beta$ and $s \leq 2^{\lfloor \log_2(p) \rfloor}$, we have $b^1 \geq b^2 \geq \beta/\sqrt{s\log_2(2p)} =: b_*$. Denote
		\begin{align*}
			 u := \frac{\ell_0\beta^2}{80s\log_2(2p)}=rb_*^2 \quad \text{and} \quad \delta := \frac{a}{2\sqrt{r+\ell}}.
		\end{align*}
		Now define the following events:
		\begin{align*}
			\Omega_{0} &:= \{z < N \leq z+r \} \\	
			\Omega_1 &:= \bigl\{  t_{N,b}^j \leq N-z+ub^{-2}\text{ for all } j\in [p] \text{ and } b\in\mathcal{B}\cup\mathcal{B}_0\bigr\}, \\
			\Omega_2 &:= \biggl\{ \biggl|A_{N,b}^{j',j}+\sum_{i=N+1}^{N+\ell} X_i^{j'}\biggr| < a\sqrt{t_{N,b}^j+\ell} \text{ for all } b\in\mathcal{B}\cup\mathcal{B}_0, j \in [p] \\
			&\hspace{5cm} \text{ and all } j'\in [p]\setminus \{j\} \text{ with } \bigl|\theta^{j'}\bigr| \leq \delta \biggr\}, \\
			\Omega_3 &:= \bigl\{t_{N,\hat{b}}^{\hat{j}} \leq  N-z+\ell/20 \bigr\}.
		\end{align*}
		Finally, we denote event
		\[
		\Omega_4 := \Omega_{4,1} \cup \Omega_{4,2},
		\]
		with
		\begin{align*}
			\Omega_{4,1} &:= \bigl\{ \hat{j} \neq 1, 1 \in \hat{\mathcal{S}}, \tilde{b}^1 \geq b^1/\sqrt{2} \bigr\} \\
			\Omega_{4,2} &:= \bigl\{ \hat{j} = 1, 2 \in \hat{\mathcal{S}}, \tilde{b}^2 \geq b^2/\sqrt{2} \bigr\}.
		\end{align*}
        We note that for $j \in \{1,2\}$, on $\Omega_{4,j}$, we have $\tilde b^j \in \mathcal{B}\cup\mathcal{B}_0$.
		Then, on the event $\bigcap_{k=0}^4 \Omega_k$, we have
		\begin{align*}
			L &= \min_{j \in \hat{\mathcal{S}}} \Bigl\{t_{N,\tilde{b}^j}^j + \frac{d_2}{ (\tilde{b}^j)^2}\Bigr\} \wedge N  \leq N-z + \frac{2(u + d_2)}{(b^2)^2} \leq 3r + \frac{2d_2}{b_*^2} \leq 8r.
		\end{align*}
		Thus, it suffices to control the probability of $\cup_{k=0}^4\Omega_k^{\mathrm{c}}$. First, we note that 
		\begin{align} \label{Eq:gensig_randomN_omega0}
			\mathbb{P}(\Omega_{0}^{\mathrm{c}}) \leq g(r;N) + \mathbb{P}(N\leq z).
		\end{align}
		On $\Omega_0$, we have for any $j \in [p]$ and $b \in \mathcal{B} \cup \mathcal{B}_0$ that
		\begin{align*}
			t_{N,b}^j &= \sargmax_{0\leq h\leq N} \sum_{i=N-h+1}^N b(X_i^j -b/2) \leq \sargmax_{N-z \leq h\leq N} \sum_{i=N-h+1}^N b(X_i^j -b/2) \\
			&= N-z+\sargmax_{0 \leq h\leq z} \sum_{i=z-h+1}^z b(X_i^j -b/2).
		\end{align*}
		Thus, by Lemma~\ref{Lemma:1DLemma} (taking $\mu = -b/2$) and a union bound, we have
		\begin{align} 
			&\mathbb{P}(\Omega_{0} \cap \Omega_{1}^{\mathrm{c}}) \leq 2p\log_2(4p)e^{-u/8} = 2p\log_2(4p) e^{-r b_*^2/8}.  \label{Eq:gensig_randomN_omega1}
		\end{align}
		Now observe that, for all $z < n \leq z+r, b\in\mathcal{B}\cup\mathcal{B}_0, j \in [p]$ and $j'\in [p]\setminus \{j\}$, we have
		\begin{equation}
		\label{Eq:Orange1}
		A_{n,b}^{j',j}+\sum_{i=n+1}^{n+\ell} X_i^{j'} \biggm| t_{n.b}^j \sim \mathcal{N}\Bigl(\theta^{j'}\bigl\{\ell + \min\bigl(t_{n,b}^j,n-z\bigr)\bigr\},t_{n,b}^j+\ell\Bigr).
		\end{equation}
		Hence, when $\bigl|\theta^{j'}\bigr| \leq \delta$, we have that
		\begin{align}
			&\mathbb{P}\biggl(  \biggl|A_{n,b}^{j',j}+\sum_{i=n+1}^{n+\ell} X_i^{j'}\biggr| \geq a\sqrt{t_{n,b}^j+\ell} \biggm | t_{n.b}^j \biggr) \leq \mathbb{P}(|Y_1| \geq a) \leq 2\mathbb{P}(Y_1 \geq a) \leq e^{-a^2/8}, \label{Eq:Orange2}
		\end{align}
		where $Y_1 \sim \mathcal{N}(\delta\sqrt{n-z+\ell}, 1)$, and where the last inequality follows from the relation $a = 2\delta\sqrt{r+\ell}$.  Thus, by a union bound, we have
		\begin{align}
			&\mathbb{P}(\Omega_{0} \cap \Omega_{2}^{\mathrm{c}}) \leq 2rp^2\log_2(4p)e^{-a^2/8}.   \label{Eq:gensig_randomN_omega2} 
		\end{align}
		Recall that $u=\ell_0b_*^2/80$. We therefore have for any $z< n \leq z+r, j \in [p]$ and $b \in \mathcal{B}$ that 
		\begin{align}
			&\mathbb{P}\Bigl( \{N = n\} \cap \Omega_{1} \cap \Omega_{2} \cap \bigl\{ Q_{n,b}^j \geq Q_{n,b_*}^j  \bigr\} \cap \bigl\{t_{n,b}^j > n-z+\ell/20\bigr\} \Bigm | X_1^j, X_2^j, \ldots \Bigr) \nonumber\\
			&\leq \mathbb{P}\Biggl( \bigcup_{ \substack{ j' \in [p]\setminus \{j\}: \\ |\theta^{j'}| > \delta }  } \! \! \! \bigl\{ |E_{n, b}^{j', j}| \geq |E_{n, b_*}^{j', j}|  \bigr\} \cap \{N=n\} \cap \Omega_{1} \cap \Omega_{2}  \cap \bigl\{t_{n,b}^j > n-z+\ell/20\bigr\} \! \Biggm| \!X_1^j, X_2^j, \ldots \! \Biggr) \nonumber\\
			&\leq \smash{\sum_{\substack{ j' \in [p]\setminus \{j\}: \\ |\theta^{j'}| > \delta}}} \mathbb{P} \Bigl(  \bigl\{ |E_{n, b}^{j', j}| \geq |E_{n, b_*}^{j', j}|  \bigr\} \cap  \bigl\{n-z < t_{n,b_*}^j \leq n-z+\ell/80 \bigr\} \nonumber\\ 
			&\hspace{9.25cm}\cap \bigl\{t_{n,b}^j > n-z+\ell/20\bigr\} \Bigm| X_1^j, X_2^j, \ldots  \Bigr)      \nonumber\\ 
			&\hspace{0.4cm} + \sum_{\substack{ j' \in [p]\setminus \{j\}: \\ |\theta^{j'}| > \delta}} \mathbb{P} \Bigl(  \bigl\{ |E_{n, b}^{j', j}| \geq |E_{n, b_*}^{j', j}|  \bigr\} \cap  \bigl\{t_{n,b_*}^j \leq n-z \bigr\} \cap \bigl\{t_{n,b}^j > n-z+\ell/20\bigr\} \Bigm| X_1^j, X_2^j, \ldots  \Bigr)     \nonumber \\
			& \leq p\exp\biggl(  -\frac{\ell \delta^2}{960}  \biggr) = p\exp\biggl(  -\frac{ \ell a^2}{3840(r+\ell)}  \biggr) , \label{Eq:Orange3}
		\end{align}
		where the final inequality follows from Lemma~\ref{Lemma:normalcomparison}(a), applied with $U=\sum_{i=n-t_{n,b_*}^j+1}^z X_i^{j'}$, $V=\sum_{i=n-t_{n,b}^j+1}^{n-t_{n,b_*}^j} X_i^{j'}$, $Y = \sum_{i=z+1}^{n+\ell} X_i^{j'}$, $\alpha = \theta^{j'}$, $\phi_1 = z-n+t_{n,b_*}^j$, $\phi_2 = z-n+t_{n,b}^j$, $\phi_3 = n-z+\ell$ and $\kappa = \ell/80$, as well as Lemma~\ref{Lemma:normalcomparison}(b), with $U=\sum_{i=n-t_{n,b}^j+1}^z X_i^{j'}$, $Y = \sum_{i=n-t_{n, b_*}^j+1}^{n+\ell} X_i^{j'}$, $Z = \sum_{i=z+1}^{n-t_{n, b_*}^j} X_i^{j'}$,  $\alpha = \theta^{j'}$, $\phi_1 = z-n+t_{n,b}^j$, $\phi_3 = \ell+t_{n,b_*}^j$, $\phi_4 = n-z-t_{n,b_*}^j$ and $\kappa = \ell/80$.  
		Observe that $Q_{n,\hat{b}}^{\hat{j}} \geq Q_{n,b_*}^{\hat{j}}$.  Thus, by a union bound, we have
		\begin{align} \label{Eq:gensig_randomN_omega3}
			\mathbb{P}(\Omega_{0} &\cap \Omega_{1} \cap \Omega_{2} \cap \Omega_{3}^{\mathrm{c}}) \nonumber \\
			&\leq \sum_{j =1}^p \sum_{b \in \mathcal{B}} \sum_{n=z+1}^{z+ \lfloor r\rfloor }  \mathbb{P}\Bigl( \{N=n\} \cap \Omega_{1} \cap \Omega_{2} \cap \bigl\{ Q_{n,b}^j \geq Q_{n,b_*}^j  \bigr\} \cap \bigl\{t_{n,b}^j > n-z+\ell/20\bigr\}\Bigr) \nonumber \\
			&\leq 2rp^2\log_2(2p) \exp\biggl(  -\frac{ \ell a^2}{3840(r+\ell)} \biggr) \leq  2rp^2\log_2(2p) e^{-a^2/3888}, 
		\end{align}
		where the last inequality follows from $r = \ell_0/80 \leq \ell/80$.
		Recall that $d_1^2=5rb_*^2/9\leq (b^1)^2\ell/144$. Thus, on $\Omega_0 \cap \Omega_3$, we have
		\begin{equation*}
			t_{N,\hat{b}}^{\hat{j}} \leq N-z+\ell/20 \leq r+\ell/20 \leq \ell/16.
		\end{equation*}
		Hence, for any $z < n \leq z+r, j \in [p]\setminus\{1\}$ and $b \in \mathcal{B}$, we have
		\begin{align}
			\mathbb{P}\bigl(\Omega_3 &\cap \{N = n, \hat{j} = j, \hat{b} = b\} \cap \Omega_{4,1}^{\mathrm{c}} \bigm| X_1^j,X_2^j,\ldots \bigl) \nonumber\\
			&\leq  \mathbb{P}\biggl( \{ t_{n, b}^j \leq \ell/16 \}  \cap \Bigl\{  E_{n, b}^{1,j} - b^1 \sqrt{\bigl(t_{n, b}^{j} + \ell\bigr)/2} < d_1\Bigr\} \biggm| X_1^j, X_2^j, \ldots \biggr) \leq \frac{1}{2}e^{-d_1^2/2}. \label{Eq:Orange4}
		\end{align}
		Here, in the final bound, we have used the facts that $E_{n,b}^{1,j}\mid  t_{n,b}^j \sim \mathcal{N}\bigl(\theta^{1} \min\{(n+\ell-z) (t_{n, b}^{j} + \ell)^{-1/2},(t_{n, b}^{j} + \ell)^{1/2}\} , 1\bigr)$ and that 
		\begin{align*}
			\theta^1 \min\bigl\{(n+\ell-z) (t_{n, b}^{j} &+ \ell)^{-1/2},(t_{n, b}^{j} + \ell)^{1/2}\bigr\} - b^1\sqrt{\bigl(t_{n,b}^j+\ell\bigr)/2} \\
			&\geq \frac{4\theta^1\sqrt{\ell}}{\sqrt{17}} - \frac{b^1\sqrt{17\ell} }{4\sqrt{2}} \geq \frac{b^1\sqrt{\ell}}{6} \geq 2d_1,
		\end{align*}
		when $t_{n, b}^j \leq \ell/16$, as well as the standard Gaussian tail bound used at the end of the proof of Lemma~\ref{Lemma:1DLemma}. By a similar argument, we also have for any $z < n \leq z+r$ and $b \in \mathcal{B}$ that
		\begin{align}
			\mathbb{P}\bigl(\Omega_3 &\cap \{N = n, \hat{j} = 1, \hat{b} = b\} \cap \Omega_{4,2}^{\mathrm{c}} \bigm| X_1^1,X_2^1,\ldots \bigl) \nonumber\\
			&\leq  \mathbb{P}\biggl( \{ t_{n, b}^1 \leq \ell/16 \}  \cap \Bigl\{  E_{n, b}^{2,1} - b^2 \sqrt{\bigl(t_{n, b}^{1} + \ell\bigr)/2} < d_1\Bigr\} \Bigm| X_1^1, X_2^1, \ldots \biggr) \leq \frac{1}{2}e^{-d_1^2/2}. \label{Eq:Orange5}
		\end{align}
		Thus, by a union bound, we have
		\begin{align} \label{Eq:gensig_randomN_omega4}
			\mathbb{P}(\Omega_{0} \cap \Omega_{3} \cap \Omega_{4}^{\mathrm{c}}) &= \mathbb{P}(\Omega_{0} \cap \Omega_{3} \cap \Omega_{4,1}^{\mathrm{c}} \cap \Omega_{4,2}^{\mathrm{c}}) \nonumber \\
			&\leq \sum_{j = 2}^p \sum_{b \in \mathcal{B} } \sum_{n=z+1}^{z+\lfloor r \rfloor} \mathbb{P}\bigl(\Omega_3 \cap \{N = n, \hat{j} = j, \hat{b} = b\} \cap \Omega_{4,1}^{\mathrm{c}} \bigr) \nonumber\\
			&\hspace{3cm} + \sum_{b \in \mathcal{B} } \sum_{n=z+1}^{z+\lfloor r \rfloor} \mathbb{P}\bigl(\Omega_3 \cap \{N = n, \hat{j} = 1, \hat{b} = b\} \cap \Omega_{4,2}^{\mathrm{c}} \bigr) \nonumber \\
			&\leq rp\log_2(2p)e^{-5rb_*^2/18}. 
		\end{align}
		Hence combining~\eqref{Eq:gensig_randomN_omega0},~\eqref{Eq:gensig_randomN_omega1},~\eqref{Eq:gensig_randomN_omega2},~\eqref{Eq:gensig_randomN_omega3} and~\eqref{Eq:gensig_randomN_omega4}, we conclude that
		\begin{align*}
			\mathbb{P}(L >8r) &\leq g(r;N) + \mathbb{P}(N\leq z) + 4rp^2\log_2(4p)e^{-a^2/3888} + 3rp\log_2(4p)e^{-rb_*^2/8} \leq \alpha,
		\end{align*}
		where the last inequality follows from the choice of $a$ with a sufficiently large universal constant $C$ and~\eqref{Eq:detectassumption}.
	\end{proof}

	\begin{proof}[Proof of Theorem~\ref{Thm:Support}]
	Fix $r \geq 1$ that satisfies the assumption~\eqref{Eq:detectassumption}.\\
	(a) For $j' \in \mathcal{S}_{\beta}^{\mathrm{c}}$, we have $|\theta^{j'}| < b_{\min}$, so the event $\{|\tilde{b}^{j'}| \leq |\theta^{j'}|\}$ is empty. Thus by~\eqref{Eq:CoveragePfEq1}, we have, for $n>z, j\in [p]$, $b \in \mathcal{B}$ and $j' \in  \mathcal{S}_{\beta}^{\mathrm{c}}$, that
		\[
		\mathbb{P}\bigl( \{j' \in \hat{\mathcal{S}}\} \cap \{N=n, \hat{j} = j, \hat{b} = b\}\bigr) \leq 2\bar{\Phi}(d_1).
		\]
		Hence, by a union bound, we have
		\begin{align*}
			\mathbb{P}(\hat{\mathcal{S}}  \nsubseteq &\mathcal{S}_{\beta}) \leq \mathbb{P}(N \leq z) + \mathbb{P}(N > z+r) \\
			&\hspace{3.5cm}+ \sum_{n=z+1}^{z+\lfloor r \rfloor}  \sum_{j=1}^p \sum_{b \in\mathcal{B}} \sum_{j'\in \mathcal{S}_{\beta}^{\mathrm{c}}}  \mathbb{P}\bigl( \{j' \in \hat{\mathcal{S}}\} \cap \{N=n, \hat{j} = j, \hat{b} = b\}\bigr) \\
			&\leq \mathbb{P}(N \leq z) + g(r;N) + 4rp^2\log_2(2p) \bar{\Phi}(d_1) \leq \alpha,
		\end{align*}
		as required.
		
		\medskip
		
		\noindent (b) 
		We use the events $\Omega_0, \Omega_1, \Omega_2, \Omega_3$ defined in the proof of Theorem~\ref{Thm:Length}.  Recall from the argument immediately below~\eqref{Eq:gensig_randomN_omega3} that we have $t_{N,\hat{b}}^{\hat{j}} \leq \ell/16$ and $d_1\leq \min_{j'\in\mathcal{S}}|\theta^{j'}|\sqrt{\ell}/12$ on $\Omega_0 \cap \Omega_3$.  Recall also the definition of $E_{n,b}^{j',j}$ from Algorithm~\ref{Alg:CI}. Then, for any $z < n \leq z+r$, $j \in [p]$, $j' \in \mathcal{S} \setminus \{j\}$ and $b \in \mathcal{B}$, we have
		\begin{align} \label{Eq:NewEvent_Support}
			\mathbb{P}&\bigl(\Omega_3 \cap \{N=n,\hat{j}=j,\hat{b}=b,j' \notin \hat{\mathcal{S}}\} \bigm| X_1^j,X_2^j,\ldots\bigr) \nonumber\\
			&=\mathbb{P}\Bigl(\Omega_3\cap\{N=n,\hat j = j,\hat b = b\}\cap \bigl\{|E_{n,b}^{j',j}| < b_{\min}\sqrt{t_{n,b}^j+\ell} + d_1\bigr\}\mid X_1^j,X_2^j,\ldots\Bigr) \nonumber\\
			&\leq \mathbb{P}\biggl( \{ t_{n, b}^j \leq \ell/16 \}  \cap \Bigl\{  |E_{n, b}^{j',j}| - b_{\min} \sqrt{t_{n, b}^{j} + \ell} < d_1\Bigr\} \biggm| X_1^j, X_2^j, \ldots \biggr) \leq \frac{1}{2}e^{-d_1^2/2},
		\end{align}
		where, in the final bound, we have used the facts that $E_{n,b}^{j',j}\mid  t_{n,b}^j \sim \mathcal{N}\bigl(\theta^{j'} \min\{(n+\ell-z) (t_{n, b}^{j} + \ell)^{-1/2},(t_{n, b}^{j} + \ell)^{1/2}\} , 1\bigr)$ and that 
		\begin{align*}
			|\theta^{j'}| \min\bigl\{(n+\ell-z) (t_{n, b}^{j} &+ \ell)^{-1/2},(t_{n, b}^{j} + \ell)^{1/2}\bigr\} - b_{\min}\sqrt{t_{n,b}^j+\ell} \\
			& \geq  \frac{4 |\theta^{j'}| \sqrt{\ell}}{\sqrt{17}} - \frac{b_{\min}\sqrt{17\ell} }{4\sqrt{2}} \geq \frac{|\theta^{j'}|\sqrt{\ell}}{6} \geq 2d_1,
		\end{align*}
		when $t_{n, b}^j \leq \ell/16$. Hence
		\begin{align*}
			\mathbb{P}&(\hat{\mathcal{S}}\cup \{\hat{j}\} \nsupseteq  \mathcal{S}) \\
			&\leq \mathbb{P}(\Omega_{0}^{\mathrm{c}}) + \mathbb{P}(\Omega_{0}\cap\Omega_1^{\mathrm{c}}) + \mathbb{P}(\Omega_{0}\cap\Omega_2^{\mathrm{c}}) + \mathbb{P}(\Omega_{0}\cap\Omega_1\cap\Omega_2\cap\Omega_3^{\mathrm{c}}) \\
			&\quad + \sum_{n=z+1}^{z+\lfloor r\rfloor}  \sum_{j=1}^p \sum_{b \in\mathcal{B}} \sum_{j' \in \mathcal{S} \setminus \{j\}}  \mathbb{P}\bigl( \Omega_3 \cap \{N=n, \hat{j} = j, \hat{b} = b, j'\notin \hat{\mathcal{S}} \}\bigr) \\
			&\leq g(r;N) + \mathbb{P}(N\leq z) + 4rp^2\log_2(4p)e^{-a^2/3888} + 3rp^2\log_2(4p)e^{-r\beta^2/(8s\log_2(2p))} \leq \alpha,
		\end{align*}
		where the penultimate inequality follows from~\eqref{Eq:gensig_randomN_omega0},~\eqref{Eq:gensig_randomN_omega1},~\eqref{Eq:gensig_randomN_omega2},~\eqref{Eq:gensig_randomN_omega3} and~\eqref{Eq:NewEvent_Support}, and the last inequality follows from the choice of $a$ with a sufficiently large universal constant $C$ and~\eqref{Eq:detectassumption}.
	\end{proof}
	
\begin{proof}[Proof of Proposition~\ref{Prop:LowerBound}]
Fix $N\in\mathcal{T}_{r,m}$ and $\psi\in\mathcal{J}_{N}$. We denote by $\mathbb{P}_{z,\theta}^{(n_0)}$ the restriction of $\mathbb{P}_{z,\theta}$ to the filtration $\mathcal{F}_{n_0} := \sigma(X_1,\ldots,X_{n_0})$.  Denote
\[ 
\Tilde{\Theta} := \bigl\{\theta \in \mathbb{R}^p: \theta^j \in \{0,  1/(8\sqrt{r})\}, |\mathrm{supp}(\theta)| = m \bigr\},
\]
    and let $\Tilde{\Theta}_{\mathrm{pa}} \subseteq \Tilde{\Theta}$ be an $(m/4)$-packing set with respect to the symmetric difference metric defined above, i.e.~for any $\theta, \Tilde{\theta} \in \Tilde{\Theta}_{\mathrm{pa}}$, we have $d\bigl(\mathrm{supp}(\theta), \mathrm{supp}(\tilde{\theta})\bigr) > m/4$. We also have $\mathrm{KL}\bigl(\mathbb{P}_{z,\theta}^{(z+r)}, P_{z,\tilde{\theta}}^{(z+r)}\bigr) = r\|\theta - \tilde{\theta}\|_2^2/2 \leq m/64$.

Enumerate $\Tilde{\Theta}_{\mathrm{pa}} = \bigl\{\theta_{(1)}, \theta_{(2)}, \ldots, \theta_{(|\Tilde{\Theta}_{\mathrm{pa}}|)}\bigr\}$. Let $\phi^* := \sargmin_{k \in [|\Tilde{\Theta}_{\mathrm{pa}}|]} d\bigl(\psi,\mathrm{supp}(\theta_{(k)})\bigr)$. Note that $\phi^*$ is also $\mathcal{F}_N$-measurable.  Then for any $z \in \mathbb{N}_0$, 
\begin{align} \label{Eq:minimax_reduction}
\sup_{\theta\in \Theta_{r,m}} \mathbb{E}_{z,\theta} d\bigl(\psi, \mathrm{supp}(\theta)\bigr) &\geq \frac{m}{8|\Tilde{\Theta}_{\mathrm{pa}}|} \sum_{k = 1}^{|\Tilde{\Theta}_{\mathrm{pa}}|} \mathbb{P}_{z,\theta_{(k)}} \Bigl( d\bigl(\psi, \mathrm{supp}(\theta_{(k)})\bigr) > \frac{m}{8} \Bigr) \nonumber \\
& \geq \frac{m}{8|\Tilde{\Theta}_{\mathrm{pa}}|} \sum_{k=1}^{|\Tilde{\Theta}_{\mathrm{pa}}|} \mathbb{P}_{z,\theta_{(k)}} ( \phi^* \neq k ) \nonumber \\
& = \frac{m}{8}  \biggl\{ 1 - \frac{1}{|\Tilde{\Theta}_{\mathrm{pa}}|} \sum_{k=1}^{|\Tilde{\Theta}_{\mathrm{pa}}|} \mathbb{P}_{z,\theta_{(k)}} ( \phi^* = k )  \biggr\} \nonumber \\
& \geq \frac{m}{8} \biggl\{ \frac{3}{4} - \frac{1}{|\Tilde{\Theta}_{\mathrm{pa}}|} \sum_{k=1}^{|\Tilde{\Theta}_{\mathrm{pa}}|} \mathbb{P}_{z,\theta_{(k)}} ( \phi^* = k, N \leq z+r)   \biggr\} \nonumber \\
& = \frac{m}{8}  \biggl\{ \frac{3}{4} - \frac{1}{|\Tilde{\Theta}_{\mathrm{pa}}|} \sum_{k=1}^{|\Tilde{\Theta}_{\mathrm{pa}}|} \mathbb{P}_{z,\theta_{(k)}}^{(z+r)} ( \phi^* = k,N \leq z+r) \biggr\}.
\end{align}
Now set 
\[
\tilde{\phi}^* := \left\{ \begin{array}{ll} \phi^* & \mbox{if $N \leq z+r$} \\
1 & \mbox{if $N > z+r$.} \end{array} \right.
\]
Then $\tilde{\phi}^*$ is $\mathcal{F}_{z+r}$-measurable and by Fano's inequality \citep[][Lemma 3]{Yu1997Fano}, we have
\begin{align} \label{Eq:minimax_fano}
 \frac{1}{|\Tilde{\Theta}_{\mathrm{pa}}|} \sum_{k=1}^{|\Tilde{\Theta}_{\mathrm{pa}}|} \mathbb{P}_{z,\theta_{(k)}}^{(z+r)} ( \phi^* = k,N \leq z+r) &\leq \frac{1}{|\Tilde{\Theta}_{\mathrm{pa}}|} \sum_{k =1}^{|\Tilde{\Theta}_{\mathrm{pa}}|} \mathbb{P}_{z,\theta_{(k)}}^{(z+r)} (\tilde \phi^* = k) \nonumber\\
 &\leq \frac{ \log 2 + |\Tilde{\Theta}_{\mathrm{pa}}|^{-2} \sum_{j, k = 1}^{|\Tilde{\Theta}_{\mathrm{pa}}|} \mathrm{KL}\bigl(\mathbb{P}_{z,\theta_{(j)}}^{(z+r)}, P_{z,{\theta_{(k)}}}^{(z+r)}\bigr)  }{\log |\Tilde{\Theta}_{\mathrm{pa}}| } \nonumber \\
 &\leq \frac{\log 2 + m/64}{\log |\Tilde{\Theta}_{\mathrm{pa}}| }.
\end{align}
By \citet[Lemma~4.7]{massart2007concentration}, there exists an $(m/4)$-packing set with
\begin{equation} \label{Eq:packing}
\log |\Tilde{\Theta}_{\mathrm{pa}}| \geq m/8.
\end{equation}
Combining \eqref{Eq:minimax_reduction}, \eqref{Eq:minimax_fano} and \eqref{Eq:packing}, we conclude that
\begin{align*}
    \sup_{z\in\mathbb{N}_0, \theta\in \Theta_{r,m}} \mathbb{E}_{z,\theta}\, d\bigl(\psi, \mathrm{supp}(\theta)\bigr) &\geq \frac{m}{8} \biggl( \frac{3}{4} - \frac{\log 2 + m/64}{m/8} \biggr) \\
    &\geq \frac{m}{8} \biggl( \frac{3}{4} - \frac{8\log 2}{m} - \frac{1}{8} \biggr) \geq \frac{m}{32},
\end{align*}
where we have used the assumption that $m \geq 15$ in the final inequality. 
\end{proof}

	\begin{proof}[Proof of Proposition~\ref{prop:ocdsatisfyassumption}]
Following the proof of \citet[][Theorem~1(b)]{CWS2021} up to, but not including,~(16), we have for every $j \in [p]$, $b \in \mathcal{B} \cup \mathcal{B}_0$ that
\begin{align}
&\mathbb{P}\Bigl(\max_{1\leq n\leq z} (bA_{n,b}^{j,j}-b^2t_{n,b}^j/2) \geq T^{\mathrm{diag}}\Bigr) \leq 1 - \bigl(1-e^{-T^{\mathrm{diag}}}\bigr)^z \leq ze^{-T^{\mathrm{diag}}}. \label{Eq:Prop5Eq1} 
\end{align}
It follows by a union bound that
\begin{align}
\mathbb{P}\Bigl(\max_{1\leq n\leq z}S_n^{\mathrm{diag}} \geq T^{\mathrm{diag}}\Bigr) &= \mathbb{P}\Bigl(\max_{1\leq n\leq z} \max_{j \in [p]} \max_{b \in \mathcal{B} \cup \mathcal{B}_0} (bA_{n,b}^{j,j}-b^2t_{n,b}^j/2) \geq T^{\mathrm{diag}}\Bigr) \nonumber\\
&\leq zp|\mathcal{B} \cup \mathcal{B}_0| e^{-T^{\mathrm{diag}}} \leq \frac{\alpha}{4}.\label{Eq:ocdassumption_term1}
\end{align}
Next, for every $j \in [p], j' \in [p]\setminus\{j\}, b \in \mathcal{B}$ and $n \in [z]$, we have $\Lambda_{n,b}^{j',j} \mid \tau_\sim \mathcal{N}(0, \tau_{n,b}^j)$, so $(\Xi_{n,b}^{j',j})^2\mid \tau_{n,b}^j \leq_{\mathrm{st}} \chi_1^2$.  Thus, by another union bound, we have
\begin{align}
\mathbb{P}\Bigl(\max_{n \in [z]}S_n^{\mathrm{off}} \geq T^{\mathrm{off}}\Bigr) \nonumber 
&= \mathbb{P}\biggl(\max_{n \in [z]} \max_{j \in [p]} \max_{b \in \mathcal{B}} \sum_{j' \in [p]\setminus\{j\}}(\Xi_{n,b}^{j',j})^2  \mathbbm{1}_{\{|\Xi_{n,b}^{j',j}| \geq \tilde{a}\}} \geq T^{\mathrm{off}}\biggr) \nonumber \\
&\leq \mathbb{P}\Bigl(\max_{n \in [z]} \max_{j \in [p]} \max_{b \in \mathcal{B}} \max_{j' \in [p]\setminus\{j\}} |\Xi_{n,b}^{j', j}| \geq \tilde{a}\Bigr) \nonumber \\
&\leq zp^2|\mathcal{B}| e^{-\tilde{a}^2/2} \leq \frac{\alpha}{4}. \label{Eq:ocdassumption_term1+} 
\end{align}
From~\eqref{Eq:ocdassumption_term1} and~\eqref{Eq:ocdassumption_term1+}, we deduce that 
\begin{equation}
\label{Eq:ocdassumption_term1_combined}
    \mathbb{P}(N \leq z) \leq \alpha/2.
\end{equation}
On the other hand, for a sufficiently large universal constant $C' > 0$, we have
	    \begin{align*}
	    r_1 &= \frac{C's\log_2(2p)\log\{p\gamma\alpha^{-1}(\beta^{-2}\vee 1)\}}{\beta^2} + 2 \\
	    &\geq \Bigl\{\frac{24T^{\mathrm{off}}\log_2(2p)}{\vartheta^2} \vee \frac{12\tilde{a}^2 s \log_2(2p)}{\vartheta^2}\vee \frac{24T^{\mathrm{diag}}s\log_2(2p)}{\beta^2} \Bigr\} + 2. 
	    \end{align*}
	    Thus, by Proposition~\ref{prop:ocd'_tail} and by increasing the value of $C$ if necessary, we have for every $r \geq r_1$ that
	    \begin{align} 
	        g(r; N) + 4rp^2\log_2^2(4p) e^{-r\beta^2/(8s\log_2(2p))} &\leq  5rp^2\log_2^2(4p)\exp\biggl\{ - \frac{r\beta^2 }{48s\log_2(2p)}\biggr\}\nonumber\\
            &\leq  \frac{240sp^2\log_2^3(4p)}{\beta^2}\frac{r\beta^2}{48s\log_2(2p)}\exp\biggl\{ - \frac{r\beta^2 }{48s\log_2(2p)}\biggr\}\nonumber\\
            &\leq \frac{240sp^2\log_2^3(4p)}{\beta^2}\exp\biggl\{ - \frac{r\beta^2 }{96s\log_2(2p)}\biggr\} \leq \frac{\alpha}{4},\label{Eq:ocdassumption_term3} 
	    \end{align}
	    where the penultimate inequality follows from the fact that $xe^{-x} \leq e^{-x/2}$ for $x \geq 0$.  The desired result follows by combining \eqref{Eq:ocdassumption_term1_combined} and \eqref{Eq:ocdassumption_term3}.
	\end{proof}

	\subsection{Auxiliary results}
	\label{Sec:Auxiliary}
	
	\begin{prop}
		\label{Prop:1DCI}
		Let $X_1, X_2, \ldots$ be independent random variables with $X_1, \ldots, X_z \stackrel{\mathrm{iid}}{\sim} \mathcal{N}(0, 1)$ and $X_{z+1}, X_{z+2}, \ldots \stackrel{\mathrm{iid}}{\sim} \mathcal{N}(\theta, 1)$. Assume that $0 < b \leq \theta$ and let $t_{n,b}$ be defined as in~\eqref{Eq:OneDimTail} for $n\in\mathbb{N}$. Then for any $\alpha \in (0, 1)$, and any stopping time $N$ satisfying $\mathbb{P}(N < z) \leq \alpha/2$, we have that the confidence interval
		\[  		
		\mathcal{C}_0 :=  \biggl[ N - t_{N,b}  -\frac{4\{\Phi^{-1}(1-\alpha/4)\}^2}{b^2},\, N\biggr]
		\]
		satisfies $\mathbb{P}(z \in \mathcal{C}_0) \geq 1-\alpha$.
	\end{prop}
	\begin{remark}
		We could also replace $4\{\Phi^{-1}(1-\alpha/4)\}^2/b^2$ by $8\log(2/\alpha)/b^2$ in the confidence interval construction, if we apply the final bound from Lemma~\ref{Lemma:1DLemma} in~\eqref{Eq:1DMain} of the proof below.
	\end{remark}
	\begin{proof}
		For $n\in\mathbb{N}$, define $R_{n,b}:=\max\{R_{n-1,b}+b(X_n-b/2), 0\}$, with $R_{0,b}=0$. By \citet[Lemma~2 in the supplement]{CWS2021}, we have $t_{N,b} = \min\{i: 0\leq i\leq N, R_{N-i,b}=0\} = \sargmax_{0\leq h\leq N} \sum_{i=N-h+1}^N b(X_i -b/2)$. Let $U_{n,b} := \sum_{i=z+1}^{z+n} (X_i-b/2)$ for $n \in \mathbb{N}$, with $U_{0, b} :=0$. Then $R_{n+z,b}\geq bU_{n,b}$ for all $n \in \mathbb{N}$.  Hence, for $y \in [0,\infty)$, we have
		\begin{equation} \label{Eq:1DMain}
			\mathbb{P}(N-t_{N,b}-y \geq z) \leq \mathbb{P}\biggl( \inf_{n\in\mathbb{N}_0:n\geq z+ y} R_{n,b} = 0\biggr)\leq \mathbb{P}\biggl( \inf_{n\in\mathbb{N}_0:n\geq  y} U_{n,b} \leq  0\biggr)  \leq 2 \bar{\Phi}\Bigl(\sqrt{y}(\theta-b/2)\Bigr),
		\end{equation}
		where the last inequality follows from Lemma~\ref{Lemma:1DLemma}. Thus, if we choose $y = 4\{\Phi^{-1}(1-\alpha/4)\}^2/b^2$, then we are guaranteed that $\mathbb{P}(N-t_{N,b}-y > z) \leq \alpha/2$. Combining this with the assumption that $\mathbb{P}(N < z) \leq \alpha/2$, the desired result follows.
	\end{proof}

	\begin{lemma}\label{Lemma:1DLemma}
		Let $Y_1, Y_2, \ldots \stackrel{\mathrm{iid}}{\sim} \mathcal{N}(\mu, 1)$. Define $U_n := \sum_{i=1}^n Y_i$ for $n \in \mathbb{N}_0$, and let $\xi:=\sargmin_{n \in \mathbb{N}_0} \mu U_n$. Then, for $y \in [0, \infty)$,  we have
		\[ 
		\mathbb{P}(\xi \geq y) \leq 	\mathbb{P}\Bigl(  \inf_{n\in\mathbb{N}_0:n\geq y} \mu U_n \leq 0 \Bigr)\leq 2 \bar{\Phi}\bigl(\sqrt{y}|\mu|\bigr) \leq e^{-y\mu^2/2} . 
		\]
	\end{lemma}
	\begin{proof}
		The first inequality holds since $\mu U_\xi\leq \mu U_0 = 0$.  For the second and third inequalities, we may assume without loss of generality that $\mu > 0$, since the result is clear when $\mu = 0$, and if $\mu < 0$ then the result will follow from the corresponding result with $\mu > 0$ by setting $Y_i' := -Y_i$ for $i \in \mathbb{N}$. Note that $(U_{n}-n\mu)_{n\in\mathbb{N}_0}$ is a standard Gaussian random walk starting at 0. Let $(B_t)_{t\in [0,\infty)}$ denote a standard Brownian motion starting at 0. Then, we have for any $y\in \mathbb{N}_0$ and $u > 0$ that
		\begin{equation} \label{Eq:Siegmund}
			\mathbb{P}\biggl(  \inf_{n\in\mathbb{N}_0:n\geq y} U_n \leq 0 \biggm | U_y = u \biggr) \leq\mathbb{P}\biggl\{  \inf_{t\in [y,\infty)} (B_t + t\mu) \leq 0 \biggm | B_y = u \biggr\} \leq e^{-2u\mu},
		\end{equation}
		where the final inequality follows from \citet[Proposition 2.4 and Equation (2.5)]{Siegmund1986}. Thus,
		for $y \in [0, \infty)$, we have
		\begin{align*}
			\mathbb{P}\biggl( \inf_{n\in\mathbb{N}_0:n\geq  y} U_n \leq 0  \biggr) &= \mathbb{P}\bigl( U_{\lceil y \rceil} \leq 0  \bigr) +  \mathbb{E}\biggl\{\mathbb{P}\biggl(  \inf_{n\in\mathbb{N}_0:n\geq \lceil y\rceil } U_n \leq 0 \biggm | U_{\lceil y\rceil } \biggr) \mathbbm{1}_{\{U_{\lceil y \rceil} > 0\}}\Bigr\}\\
			&\leq \bar{\Phi}\Bigl(\sqrt{\lceil y \rceil}\mu\Bigr)  + \int_{0}^{\infty}  \frac{1}{\sqrt{2\pi \lceil y \rceil}} e^{-\frac{(u-\lceil y \rceil\mu)^2}{2\lceil y \rceil}}e^{-2u\mu}\, du  \nonumber\\
			&= 2\bar{\Phi}\Bigl(\sqrt{\lceil y \rceil}\mu\Bigr) \leq 2 \bar{\Phi}\bigl(\sqrt{y}\mu\bigr) \leq e^{-y\mu^2/2},
		\end{align*}
		where the first inequality follows from \eqref{Eq:Siegmund} and the fact that $U_{\lceil y \rceil} \sim \mathcal{N}(\lceil y \rceil \mu, \lceil y\rceil)$ and the last inequality follows from the standard normal distribution tail bound $\bar{\Phi}(x) \leq e^{-x^2/2}/2$ for $x \geq 0$.
	\end{proof}
	
	In Proposition~\ref{prop:ocd'_tail}, we assume the Gaussian data generating mechanism given at the beginning of Section~\ref{Sec:Theory}, and show that for the \texttt{ocd}$'$ base procedure, the quantity $g(r;N)$ from~\eqref{Eq:grn} has essentially the same form as the final term in~\eqref{Eq:detectassumption}.
	\begin{prop} \label{prop:ocd'_tail}
		Assume that $\theta$ has an effective sparsity of $s := s(\theta) \geq 2$. Then, the output $N$ from \textup{\texttt{ocd}$'$}, with inputs $(X_t)_{t\in \mathbb{N}}$, $0 < \beta \leq \vartheta$, $\tilde{a} > 0$, $T^{\mathrm{diag}}>0$ and $T^{\mathrm{off}} >0$, satisfies
		\[
		\mathbb{P}_{z, \theta,\Sigma} \bigl(N > z +  r \bigl) \leq p \exp \biggl \{   -\frac{ \beta^2(r-1)}{24s\log_2(2p)}\biggr\},
		\]
		for all $r \geq \Bigl\{\frac{24T^{\mathrm{off}}\log_2(2p)}{\vartheta^2} \vee \frac{12\tilde{a}^2 s \log_2(2p)}{\vartheta^2}\vee \frac{24T^{\mathrm{diag}}s\log_2(2p)}{\beta^2} \Bigr\} + 2$. 
	\end{prop}
	
	\begin{proof}
		For $\theta\in\mathbb{R}^p$ with effective sparsity $s(\theta)$, there is at most one coordinate in $\theta$ of magnitude larger than $\vartheta/\sqrt{2}$, so there exists $b_* \in \bigl\{ \beta/\sqrt{s(\theta)\log_2(2p)}, -\beta/\sqrt{s(\theta)\log_2(2p)}\bigr\} \subseteq \mathcal{B}$ such that 
		\begin{equation*}
			\mathcal{J}:=\Bigl\{j \in [p]: \theta^j  / b_* \ge 1 \text{ and } |\theta^j|\le \vartheta/\sqrt{2} \Bigr\} 
		\end{equation*}
		has cardinality at least $s(\theta)/2$.  Note that the condition $\theta^j/b_*\geq 1$ above ensures that $\{\theta^j:j\in\mathcal{J}\}$ all have the same sign as $b_*$. By \citet[Proposition~8]{CWS2021}, we have on the event $\{N > z\}$ that
		\begin{equation}
			\label{Eq:q}
			q(X_1,\ldots,X_z,\theta) := 
			\max\bigl\{t_{z,b_*}^j: j \in \mathcal{J}\bigr\} \leq \frac{8T^{\mathrm{diag}}s\log_2(2p)}{\beta^2}.
		\end{equation}
		We now fix 
		\begin{equation} \label{Eq:AssumptionOnr}
			r \geq \biggl\{\frac{24T^{\mathrm{off}}\log_2(2p)}{\vartheta^2} \vee \frac{12\tilde{a}^2 s \log_2(2p)}{\vartheta^2}\vee \frac{24T^{\mathrm{diag}}s\log_2(2p)}{\beta^2} \biggr\} + 2 =: r_0.
		\end{equation} 
		For $j \in \mathcal{J}$, define the event
		\[
		\Omega_r^j := \bigl\{t_{z + \lfloor r \rfloor, b_*}^{j} > 2\lfloor r \rfloor/3\bigr\}.
		\]
		By applying \citet[Lemma~2]{CWS2021} to $t_{z + \lfloor r \rfloor, b_*}^j$, we have for $j \in \mathcal{J}$ that
		\begin{align*}
			t_{z + \lfloor r \rfloor, b_*}^j &= \sargmax_{0 \leq h \leq z + \lfloor r \rfloor} \sum_{i=z + \lfloor r \rfloor -h+1}^{z + \lfloor r \rfloor} b_*(X_i^j - b_*/2) \geq \sargmax_{0 \leq h \leq \lfloor r \rfloor} \sum_{i=z + \lfloor r \rfloor -h+1}^{z + \lfloor r \rfloor} b_*(X_i^j - b_*/2) \\
			&= \sargmax_{0 \leq h \leq \lfloor r \rfloor} \sum_{i=z+1}^{z + \lfloor r \rfloor-h} -b_*(X_i^j - b_*/2) = \lfloor r \rfloor - \largmax_{0 \leq h \leq \lfloor r \rfloor} \sum_{i=z+1}^{z + h} -b_*(X_i^j - b_*/2).
		\end{align*}
		Recall that $X_{z+1}^j, X_{z+2}^j, \ldots \stackrel{\mathrm{iid}}{\sim} \mathcal{N}(\theta^j, 1)$. 
		Hence, by applying Lemma~\ref{Lemma:1DLemma} with $\mu = |b_*|/2$ and $y = \lfloor r \rfloor / 3$, we have for each $j \in \mathcal{J}$ that
		\begin{align}
			\mathbb{P}\bigl\{(\Omega_{r}^j)^{\mathrm{c}}\bigr\} &=  \mathbb{P}\biggl( t_{z + \lfloor r \rfloor, b_*}^j \leq \frac{2\lfloor r \rfloor}{3} \biggr) \nonumber\\
			&\leq   \mathbb{P}\biggl( \largmax_{0 \leq h \leq \lfloor r \rfloor} \sum_{i=z+1}^{z + h} -b_*(X_i^j - b_*/2) \geq \frac{\lfloor r \rfloor}{3} \biggr) \nonumber \\
			&\leq   \mathbb{P}\biggl( \sup_{h \geq \lfloor r \rfloor/3} \sum_{i=z+1}^{z + h} - \mathrm{sgn}(b_*)(X_i^j - b_*/2) \geq 0 \biggr) \leq \exp\bigl(-b_*^2\lfloor r \rfloor/24  \bigr). \label{Eq:SmallProbShortTail} 
		\end{align}   
		We now work on the event $\Omega_{r}^{j}$, for some fixed $j \in \mathcal{J}$. We note that~\eqref{Eq:AssumptionOnr} guarantees that $ r  \ge 2$, and thus $t_{z + \lfloor r \rfloor, b_*}^{j} \geq \bigl\lceil2\lfloor r\rfloor/3\bigr\rceil \ge 2$. Then, by~\eqref{Eq:q} and~\eqref{Eq:AssumptionOnr}, we have $r_0>3t_{z, b_*}^j$, and hence by \citet[Lemma~9]{CWS2021},
		\[
		\frac{\lfloor r \rfloor}{3} < \frac{t_{z + \lfloor r \rfloor, b_*}^{j}}2  \leq \tau_{z + \lfloor r \rfloor, b_*}^{j} \le \frac{3t_{z + \lfloor r \rfloor, b_*}^{j}}{4} \le \frac{3\bigl(t_{z, b_*}^{j} +   r \bigr)}{4} <   r.
		\]
		We conclude that 
		\begin{equation}  \label{Eq:AuxTailBounds}
			2/3 \le \lfloor r \rfloor/3 < \tau_{z + \lfloor r \rfloor, b_*}^{j} \leq \lfloor r \rfloor.
		\end{equation}
		Recall that $\Lambda_{z+\lfloor r \rfloor, b_*}^{\bdot, j} \in \mathbb{R}^p$ records the tail CUSUM statistics with tail length $\tau_{z + \lfloor r \rfloor, b_*}^{j}$. We observe by~\eqref{Eq:AuxTailBounds} that only post-change observations are included in $\Lambda_{z+\lfloor r \rfloor, b_*}^{\bdot, j}$. Hence we have that
		\begin{equation} 
			\label{Eq:Lambdarj}
			\Lambda_{z + \lfloor r \rfloor, b_*}^{j', j} \bigm | \tau_{z + \lfloor r \rfloor, b_*}^{j} \stackrel{\mathrm{ind}}{\sim} \mathcal{N}\bigl(\theta^{j'} \tau_{z + \lfloor r \rfloor, b_*}^{j},  \tau_{z + \lfloor r \rfloor, b_*}^{j}\bigr)
		\end{equation}
		for $j' \in [p]\setminus\{j\}$. 
		By the definition of the effective sparsity of $\theta$, the set
		\begin{equation*}
			{\mathcal{L}}^{j}:=\biggl\{  j' \in [p] \setminus [j]: |\theta^{j'}| \geq \frac{\vartheta}{\sqrt{s\log_2(2p)}}\biggr \}
		\end{equation*}
		has cardinality at least $s-1$.  Hence, by~\eqref{Eq:AuxTailBounds}, for all $j' \in	{\mathcal{L}}^{j} $,
		\[
		|\theta^{j'}| \sqrt{\tau_{z + \lfloor r \rfloor, b_*}^{j}} > \sqrt{\frac{\vartheta^2\lfloor r \rfloor}{3s\log_2(2p)}}=:\tilde{a}_r.
		\]
		We then observe, from~\eqref{Eq:AssumptionOnr}, that 
		\begin{equation} \label{Eq:tildea_r}
			\tilde{a}_r > 2\tilde{a}.
		\end{equation} 
		Hence, from~\eqref{Eq:Lambdarj}, we have for all $j' \in {\mathcal{L}}^{j}$ that
		\begin{align}
			&\mathbb{P}\biggl(\Omega_{r}^{j}  \cap \biggl\{|\Lambda^{j', j}_{z+\lfloor r \rfloor, b_*}| < \frac{1}{2} \tilde{a}_r \sqrt{\tau^{j}_{z+\lfloor r \rfloor, b_*}}\biggr\}  \biggm | \tau_{z + \lfloor r \rfloor, b_*}^{j} \biggr) \leq \frac{1}{2}e^{-\tilde{a}_r^2/8} . \label{Eq:Orange6}
		\end{align}
		We denote
		\[
		U^j := \bigcap_{j'\in {\mathcal{L}}^{j}} \biggl\{|\Lambda^{j', j}_{z+\lfloor r \rfloor, b_*}| \geq \frac{1}{2} \tilde{a}_r \sqrt{\tau^{j}_{z+\lfloor r \rfloor, b_*}}\biggr\}.
		\]
		Thus, by a union bound, we have
		\begin{align}
		    &\mathbb{P}\bigl(\Omega_r^j \cap (U^j)^{\mathrm{c}}\bigr) \leq \frac{p}{2}e^{-\tilde{a}_r^2/8}. \label{Eq:Allcoordinates}
		\end{align}
		Moreover, on the event $\Omega_r^j \cap U^j$, we have
		\begin{align}  \label{Eq:AboveThreshold}
                  \sum_{j' \in [p]: j' \neq j}\frac{\bigl(\Lambda^{j', j}_{z+\lfloor r \rfloor, b_*}\bigr)^2}{\tau^j_{z+\lfloor r \rfloor, b_*} \vee 1} \mathbbm{1}_{\bigl\{|\Lambda^{j', j}_{z+\lfloor r \rfloor, b_*}| \geq \tilde{a}\sqrt{\tau^j_{z+\lfloor r \rfloor, b_*}}\bigr\}} 
                  &\geq  \sum_{j' \in \mathcal{L}^j}\frac{\bigl(\Lambda^{j', j}_{z+\lfloor r \rfloor, b_*}\bigr)^2}{\tau^j_{z+\lfloor r \rfloor, b_*}}  \nonumber \\
                  &\geq \frac{\tilde{a}_r^2}{4} |{\mathcal{L}}^{j}| \geq \frac{\vartheta^2\lfloor r \rfloor}{24\log_2(2p)} \geq T^{\mathrm{off}},
		\end{align}
		where the penultimate inequality uses the fact that $|{\mathcal{L}}^{j}| \geq s-1$ and the last inequality follows from~\eqref{Eq:AssumptionOnr}. We now denote
		\[
		\tilde{E}_r^j := \Biggl\{    \sum_{j' \in [p]: j' \neq j}\frac{\bigl(\Lambda^{j', j}_{z+\lfloor r \rfloor, b_*}\bigr)^2}{\tau^j_{z+\lfloor r \rfloor, b_*} \vee 1} \mathbbm{1}_{\bigl\{|\Lambda^{j', j}_{z+\lfloor r \rfloor, b_*}| \geq \tilde{a}\sqrt{\tau^j_{z+\lfloor r \rfloor, b_*}}\bigr\}}   <  T^{\mathrm{off}}   \Biggr \}.
		\]
		By~\eqref{Eq:AboveThreshold}, we have $\Omega_r^j\cap \tilde E_r^j  \subseteq \Omega_r^j \cap (U^j)^{\mathrm{c}}$. Thus, by~\eqref{Eq:SmallProbShortTail} and \eqref{Eq:Allcoordinates} we have that
		\begin{align*}
			\mathbb{P}\bigl(N > z +  r \bigr) &\leq \mathbb{P}\bigl(N > z + \lfloor r \rfloor \bigr) \leq  \mathbb{P}\biggl( \bigcap_{j \in \mathcal{J}}\tilde E_r^j\biggr) \leq \min_{j\in\mathcal{J}} \mathbb{P}(\tilde E_r^j)  \\
			&\leq\min_{j\in\mathcal{J}}\bigl\{ \mathbb{P}\bigl((U^j)^{\mathrm{c}} \cap \Omega_r^j\bigr) + \mathbb{P}\bigl((\Omega_r^j)^{\mathrm{c}}\bigr)\bigr\}  \\
			&\leq  \frac{p}{2}\exp\biggl\{ - \frac{\vartheta^2(r-1) }{24s\log_2(2p)}\biggr\}  + \exp \biggl \{   -\frac{ \beta^2(r-1)}{24s\log_2(2p)} \biggr\} \\
			&\leq  p \exp \biggl \{   -\frac{ \beta^2(r-1)}{24s\log_2(2p)} \biggr\} . 
		\end{align*}
		as desired.  
	\end{proof}
	
	
	\begin{lemma} \label{Lemma:normalcomparison}
		Let $U \sim \mathcal{N}(0, \phi_1)$, $V \sim \mathcal{N}(0, \phi_2-\phi_1)$, $Y \sim \mathcal{N}(\alpha \phi_3, \phi_3)$ and $Z \sim \mathcal{N}(\alpha \phi_4, \phi_4)$ be independent random variables.
		\begin{enumerate}[label=(\alph*)]
			\item Assume that $\min\{\phi_2, \phi_3\}/4 \geq \kappa \geq \phi_1 \geq 0$ for some $\kappa > 0$. Then
			\begin{align*}
				&\mathbb{P}\biggl(\frac{|U+V+Y|}{\sqrt{\phi_2+\phi_3}}  \geq \frac{|U+Y|}{\sqrt{\phi_1+\phi_3}} \biggr) \leq \exp\biggl(-\frac{\kappa\alpha^2}{6}\biggr).
			\end{align*}  
			\item Assume that $\min\{\phi_1, \phi_3\}/4 \geq \kappa \geq \phi_4 \geq 0$ for some $\kappa > 0$. Then
			\begin{align*}
				&\mathbb{P}\biggl(\frac{|U+Y+Z|}{\sqrt{\phi_1+\phi_3+\phi_4}} \geq \frac{|Y|}{\sqrt{\phi_3}} \biggr) \leq \exp\biggl(-\frac{\kappa\alpha^2}{12}\biggr).
			\end{align*}
		\end{enumerate}
	\end{lemma}
	\begin{proof}
		The case $\alpha = 0$ is trivial in both cases, so without loss of generality, we may assume $\alpha > 0$ throughout the rest of the proof. \\
		\noindent (a) Let
		\[ 
		W_1 := \bigl(\sqrt{\phi_2+\phi_3}- \sqrt{\phi_1+\phi_3}\bigr)(U+Y)-\sqrt{\phi_1+\phi_3} \, V, 
		\]
		so that
		\begin{align*}
			W_1 \sim \mathcal{N} \Bigl( \alpha\phi_3 \bigl( \sqrt{\phi_2+\phi_3} - \sqrt{\phi_1+\phi_3} \bigl), \bigl\{ \bigl(\sqrt{\phi_2+\phi_3}- \sqrt{\phi_1+\phi_3}\bigr)^2 + \phi_2-\phi_1\bigr\} (\phi_1+\phi_3)  \Bigr).
		\end{align*}
		Hence, by the standard Gaussian tail bound used at the end of the proof of Lemma~\ref{Lemma:1DLemma}, we have
		\begin{align} 
			&\mathbb{P}(W_1 \leq 0) \leq \frac{1}{2}e^{-\alpha^2/(2w_1)}, \label{Eq:Normal_W1}
		\end{align}
		where $w_1 := \frac{\phi_1+\phi_3}{\phi_3^2} \bigl(1 + \frac{\phi_2-\phi_1}{  (\sqrt{\phi_2+\phi_3}- \sqrt{\phi_1+\phi_3})^2}\bigr)$. Then
		\begin{align}
			\label{Eq:W1bound}
			w_1 &= \frac{\phi_1+\phi_3}{\phi_3^2} \Biggl(1 + \frac{  \bigl(\sqrt{\phi_2+\phi_3} + \sqrt{\phi_1+\phi_3}\bigr)^2 }{\phi_2-\phi_1}\Biggr) \nonumber \\
			&\leq \frac{5}{16\kappa}\biggl( 1 + \frac{ \bigl(\sqrt{8\kappa} + \sqrt{5\kappa}\bigr)^2}{3\kappa} \biggr) \leq \frac{3}{\kappa},
		\end{align}
		where the first inequality holds because $w_1$ is increasing in $\phi_1$ and decreasing in both $\phi_2$ and $\phi_3$.  Hence, using the fact that $-(U+V+Y) \leq_{\mathrm{st}} U+V+Y$, as well as~\eqref{Eq:Normal_W1} and~\eqref{Eq:W1bound}, we have
		\begin{align*}
			\mathbb{P}\biggl( \frac{|U+V+Y|}{\sqrt{\phi_2+\phi_3}} \geq \frac{|U+Y|}{\sqrt{\phi_1+\phi_3}} \biggr) &\leq \mathbb{P}\biggl(\biggl\{\frac{U+Y}{\sqrt{\phi_1+\phi_3}} \leq \frac{U+V+Y}{\sqrt{\phi_2+\phi_3}}\biggr\} \cap \{U+V+Y \geq 0\}\biggr) \nonumber\\
			&\hspace{0.2cm} + \mathbb{P}\biggl(\biggl\{\frac{U+Y}{\sqrt{\phi_1+\phi_3}} \leq -\frac{U+V+Y}{\sqrt{\phi_2+\phi_3}}\biggr\} \cap \{U+V+Y < 0\}\biggr) \nonumber\\
			&\leq 2\mathbb{P}\biggl(\frac{U+Y}{\sqrt{\phi_1+\phi_3}} \leq \frac{U+V+Y}{\sqrt{\phi_2+\phi_3}}\biggr) \nonumber\\
			&= 2\mathbb{P}(W_1\leq 0) \leq \exp\biggl(-\frac{\kappa\alpha^2}{6}\biggr),
		\end{align*}
		as required.
		
		\medskip
		
		\noindent (b) Let
		\[ 
		W_2 := \bigl(\sqrt{\phi_1+\phi_3+\phi_4}- \sqrt{\phi_3}\bigr)Y-\sqrt{\phi_3}(U+Z), 
		\]
		so that
		\begin{align*}
			W_2 \sim \mathcal{N} \Bigl(  \alpha \phi_3 \sqrt{\phi_1\!+\!\phi_3\!+\!\phi_4} - \alpha(\phi_3+\phi_4)\sqrt{\phi_3},   \bigl\{\bigl(\sqrt{\phi_1\!+\!\phi_3\!+\!\phi_4}- \sqrt{\phi_3}\bigr)^2+\phi_1+\phi_4 \bigr\}\phi_3\Bigr).
		\end{align*}
		Note that the assumption guarantees that $\mathbb{E}(W_2) > 0$. Hence, by the standard Gaussian tail bound used at the end of the proof of Lemma~\ref{Lemma:1DLemma}, we have
		\begin{align} 
			&\mathbb{P}(W_2 \leq 0) \leq \frac{1}{2}e^{-\alpha^2/(2w_2)},    \label{Eq:Normal_W2}
		\end{align}
		where 
		\begin{align*}
			w_2 &:= \frac{(\sqrt{\phi_1+\phi_3+\phi_4}-\sqrt{\phi_3})^2+\phi_1+\phi_4}{\bigl(\sqrt{\phi_3(\phi_1+\phi_3+\phi_4)}-\phi_3-\phi_4\bigr)^2}.
		\end{align*}
		Calculating the partial derivatives of $w_2$ with respect to $\phi_1, \phi_3$ and $\phi_4$ and simplifying the expressions, we have
		\begin{align*}
			\frac{\partial w_2}{\partial \phi_1} &= \frac{(\phi_3+\phi_4)\sqrt{\phi_3}-(\phi_3+2\phi_4) \sqrt{\phi_1+\phi_3+\phi_4} }{ \sqrt{\phi_1+\phi_3+\phi_4} \bigl(\sqrt{\phi_3(\phi_1+\phi_3+\phi_4)}-\phi_3-\phi_4\bigr)^3  } \leq 0, \\
			\frac{\partial w_2}{\partial \phi_3} &= \frac{-\bigl(\sqrt{\phi_1+\phi_3+\phi_4}-\sqrt{\phi_3}\bigr)^2 \bigl[ 3\phi_1+\phi_4+\bigl(\sqrt{\phi_1+\phi_3+\phi_4}-\sqrt{\phi_3}\bigr)^2 \bigr]  }{ 2\sqrt{\phi_3(\phi_1+\phi_3+\phi_4)} \bigl(\sqrt{\phi_3(\phi_1+\phi_3+\phi_4)}-\phi_3-\phi_4\bigr)^3  } \leq 0, \\
			\frac{\partial w_2}{\partial \phi_4} &= \frac{ 2\phi_1\bigl(2\sqrt{\phi_1+\phi_3+\phi_4}-\sqrt{\phi_3}\bigr)^2 + 3(\phi_3+\phi_4)\bigl( \sqrt{\phi_1+\phi_3+\phi_4}-\sqrt{\phi_3}\bigr)^2  }{ 2(\phi_1+\phi_3+\phi_4) \bigl(\sqrt{\phi_3(\phi_1+\phi_3+\phi_4)}-\phi_3-\phi_4\bigr)^3  } \\
			&\hspace{1cm} + \frac{(\phi_1+\phi_4)(\phi_3+\phi_4)}{2(\phi_1+\phi_3+\phi_4) \bigl(\sqrt{\phi_3(\phi_1+\phi_3+\phi_4)}-\phi_3-\phi_4\bigr)^3 } \geq 0.
		\end{align*}
		Thus $w_2$ is increasing in $\phi_4$ and decreasing in both $\phi_1$ and $\phi_3$ and hence 
		\begin{align}
			\label{Eq:W2bound}
			w_2 \leq \frac{6}{\kappa}.
		\end{align}
		Hence, using the fact that $-(U+Y+Z) \leq_{\mathrm{st}} U+Y+Z$, as well as~\eqref{Eq:Normal_W2} and~\eqref{Eq:W2bound}, we have
		\begin{align*}
			\mathbb{P}\biggl( \frac{|U+Y+Z|}{\sqrt{\phi_1+\phi_3+\phi_4}} \geq \frac{|Y|}{\sqrt{\phi_3}} \biggr) &\leq \mathbb{P}\biggl(\biggl\{\frac{Y}{\sqrt{\phi_3}} \leq \frac{U+Y+Z}{\sqrt{\phi_1+\phi_3+\phi_4}}\biggr\} \cap \{U+Y+Z \geq 0\}\biggr) \nonumber\\
			&\hspace{0.3cm} + \mathbb{P}\biggl(\biggl\{\frac{Y}{\sqrt{\phi_3}} \leq -\frac{U+Y+Z}{\sqrt{\phi_1+\phi_3+\phi_4}} \biggr\} \cap \{U+Y+Z < 0\}\biggr) \nonumber\\
			&\leq 2\mathbb{P}\biggl(\frac{Y}{\sqrt{\phi_3}} \leq \frac{U+Y+Z}{\sqrt{\phi_1+\phi_3+\phi_4}}\biggr) \nonumber\\
			&= 2\mathbb{P}(W_2\leq 0) \leq \exp\biggl(-\frac{\kappa\alpha^2}{12}\biggr),
		\end{align*}
		as required.
	\end{proof}
	
	\subsection{The \texorpdfstring{\texttt{ocd}}{\texttt{ocd}} and \texorpdfstring{\texttt{ocd}$'$}{\texttt{ocd}'} base procedures}
	\label{Sec:ocdprimealgorithm}
	Algorithms~\ref{Alg:ocd} and~\ref{Alg:ocd_variant}, which are taken from \citet{CWS2021}, provide two options for a base online changepoint detection procedure. The \texttt{ocd} algorithm is recommended for practical use, while its variant, the \texttt{ocd}$'$ algorithm, satisfies the key condition~\eqref{Eq:detectassumption} that underpins our theoretical results in Section~\ref{Sec:Theory}.
	
	In order to help make the paper self-contained, and to aid interpretability, we remark that the input parameter $\tilde{a}$ represents a thresholding level employed in the definition of the quantity $Q_b^j$ in Algorithm~\ref{Alg:ocd} and $\tilde{Q}_b^j$ in Algorithm~\ref{Alg:ocd_variant} that is designed to ensure that we only aggregate over signal coordinates.  The input parameters $T^{\mathrm{diag}}$ and $T^{\mathrm{off}}$ represent critical values for the diagonal and off-diagonal statistics $S^{\mathrm{diag}}$ and $S^{\mathrm{off}}$ respectively, both defined in these algorithms.  In other words, we declare a change as soon as either $S^{\mathrm{diag}} \geq T^{\mathrm{diag}}$ or $S^{\mathrm{off}} \geq T^{\mathrm{off}}$.
	
	\begin{algorithm}[htbp!]
	\KwIn{$X_1, X_2, \ldots \in \mathbb{R}^p$ observed sequentially, $\beta>0$, $\tilde{a}\geq 0$,  $T^{\mathrm{diag}}>0$ and $T^{\mathrm{off}}>0$}
	\KwSet{$b_{\min}= \frac{\beta}{\sqrt{2^{\lfloor \log_2 (2p)\rfloor} \log_2(2p)}}$, $\mathcal{B}_0 = \{\pm b_{\min}\}$, 
		$\mathcal{B} = \bigl\{\pm 2^{m/2} b_{\min}: m = 1,\ldots,\lfloor \log_2 (2p) \rfloor \bigr\}$, $n = 0$, $A_b = \mathbf{0} \in \mathbb{R}^{p\times p}$ and $t_b = 0\in\mathbb{R}^p$ for all $b \in \mathcal{B}\cup \mathcal{B}_0$}
	\Repeat{$S^{\mathrm{diag}} \geq T^{\mathrm{diag}}$ \textup{or} $S^{\mathrm{off}} \geq T^{\mathrm{off}}$}{
		$n \leftarrow n+1$\\
		observe new data vector $X_n$ \\
		\For{$(j, b) \in [p]\times (\mathcal{B} \cup \mathcal{B}_0)$}{
			$t^j_b \leftarrow t^j_b +1$ \\ 
			$A^{\bdot,j}_b \leftarrow A^{\bdot,j}_b + X_n$ \\
			\If{$bA^{j,j}_{b} - b^2t^j_b/2 \leq 0$}{                 
				$t^j_b \leftarrow 0$ and  $A^{\bdot,j}_b \leftarrow 0$
			} 
			compute $Q^j_b \leftarrow \sum_{j' \in [p]: j' \neq j}\frac{(A^{j', j}_b)^2}{t^j_b \vee 1} \mathbbm{1}_{\bigl\{|A^{j', j}_b| \geq \tilde{a}\sqrt{t_b^j}\bigr\}}$
		}
		$S^{\mathrm{diag}}\leftarrow \max_{(j, b) \in [p]\times (\mathcal{B} \cup \mathcal{B}_0)} \bigl(b A^{j,j}_b - b^2t^j_b/2\bigr)$\\ 
		$S^{\mathrm{off}} \leftarrow \max_{(j,b)\in[p]\times \mathcal{B}} Q_{b}^j$
	}
	\KwOut{$N=n$}
	\caption{Pseudo-code of the \texttt{ocd} algorithm}
	\label{Alg:ocd}
\end{algorithm}

	\begin{algorithm}[htbp!]
		\KwIn{$X_1, X_2, \ldots \in \mathbb{R}^p$ observed sequentially, $\beta>0$, $\tilde{a} \geq 0$, $T^{\mathrm{diag}} > 0$ and $T^{\mathrm{off}}>0$}
		\KwSet{$b_{\min}= \frac{\beta}{\sqrt{2^{\lfloor \log_2 (2p)\rfloor} \log_2(2p)}}$, $\mathcal{B}_0 = \{\pm b_{\min}\}$, 
		$\mathcal{B} = \bigl\{\pm 2^{m/2} b_{\min}: m = 1,\ldots,\lfloor \log_2 (2p) \rfloor \bigr\}$, $n=0$, $A_b = \Lambda_b = \tilde{\Lambda}_b = \mathbf{0} \in \mathbb{R}^{p\times p}$ and $t_b = \tau_b = \tilde{\tau}_b = 0\in\mathbb{R}^p$ for all $b \in \mathcal{B}\cup \mathcal{B}_0$}
		\Repeat{$S^{\mathrm{diag}} \geq T^{\mathrm{diag}}$ \textup{or} $S^{\mathrm{off}} \geq T^{\mathrm{off}}$}{
			$n \leftarrow n+1$\\
			observe new data vector $X_n$ \\
			\For{$(j, b) \in [p]\times (\mathcal{B} \cup \mathcal{B}_0)$}{
				$t^j_b\leftarrow t^j_b +1 \;\; \text{and}\;\; A^{\bdot,j}_b \leftarrow  A^{\bdot,j}_b + X_n$\\
				\medskip
				set $\delta = 0$ if $t_b^j$ is a power of 2 and $\delta = 1$ otherwise.\\
				$\tau_b^j \leftarrow \tau_b^j\delta + \tilde \tau_b^j (1-\delta) + 1 \;\; \text{and} \;\; \Lambda_b^{\bdot, j} \leftarrow \Lambda_b^{\bdot, j}\delta + \tilde \Lambda_b^{\bdot, j} (1-\delta) + X_n$\\
				$\tilde \tau_b^j \leftarrow (\tilde\tau_b^j+1)\delta \;\; \text{and} \;\; \tilde \Lambda_b^{\bdot, j} \leftarrow (\tilde \Lambda_b^{\bdot, j} + X_n) \delta.$\\
				\medskip
				\If{$bA^{j,j}_{b} - b^2t^j_b/2 \leq 0$}{                 
					$t_b^j \leftarrow \tau_b^j \leftarrow \tilde{\tau}_b^j \leftarrow 0$\\
					$A_b^{\bdot,j} \leftarrow \Lambda_b^{\bdot,j} \leftarrow \tilde{\Lambda}_b^{\bdot,j} \leftarrow 0$
				} 
				$\Xi_b^{\bdot,j} \leftarrow \Lambda_b^{\bdot, j} / (\tau_b^j\vee 1)^{1/2}$\\
				$\tilde{Q}^j_b \leftarrow \sum_{j' \in [p]\setminus\{j\}}(\Xi_b^{j',j})^2  \mathbbm{1}_{\{|\Xi_b^{j',j}| \geq \tilde{a}\}}$
			}
			$S^{\mathrm{diag}} \leftarrow \max_{(j, b) \in [p]\times (\mathcal{B} \cup \mathcal{B}_0)} \bigl(b A^{j,j}_b - b^2t^j_b/2\bigr)$\\ 
			$S^{\mathrm{off}} \leftarrow \max_{(j,b)\in[p]\times \mathcal{B}} \tilde{Q}_{b}^j$\\
		}
		\KwOut{$N = n$}
		\caption{Pseudo-code for the \texttt{ocd$'$} algorithm, a slight variant of \texttt{ocd}}
		\label{Alg:ocd_variant}
	\end{algorithm}
	
	\subsection{Results under sub-Gaussian and sub-exponential assumptions} \label{Sec:BeyondGaussianity}
	
	This section provides justification for the claimed theoretical results in Section~\ref{SubSec:Relax}.  We will rely on the following three propositions, the first of which is standard \citep[e.g.][Proposition~2.5 and~(2.18)]{Wainwright2019}.
	\begin{prop}[Hoeffding-type and Bernstein-type tail bound] \label{Prop:SubGauExpTailBound}
	(a) Let $a_1, \ldots, a_n \in \mathbb{R}$ and let $X_1, \ldots, X_n$ be independent sub-Gaussian random variables with variance parameter~$1$. Then 
	\[
	\mathbb{P}\Bigl(\sum_{i=1}^n a_iX_i \geq x\Bigr) \leq \exp \biggl(- \frac{x^2}{2\sum_{i=1}^n a_i^2} \biggr)
	\]
	for all $x \geq 0$.
	
	\noindent (b) Let $a_1, \ldots, a_n \in \mathbb{R}$ and let $X_1, \ldots, X_n$ be independent sub-exponential random variables with variance parameter $1$ and rate parameter $A > 0$. Then 
	\[
	\mathbb{P}\Bigl(\sum_{i=1}^n a_iX_i \geq x\Bigr) \leq \exp \biggl\{-\min\biggl( \frac{x^2}{2\sum_{i=1}^n a_i^2},\frac{Ax}{2 \max_{i \in [n]} |a_i|}  \biggr)\biggr\}
	\]
	for all $x \geq 0$.
	\end{prop}
	One special case where we apply this proposition frequently is with $a_1 = \ldots = a_n = n^{-1/2}$. The two bounds are then $e^{-x^2/2}$ and $e^{-x(x\wedge A\sqrt{n})/2}$ respectively.
	
	The following proposition can be used in place of Lemma~\ref{Lemma:1DLemma} to control excursion probabilities of sub-Gaussian and sub-exponential random walks with drift.
	\begin{prop}
	\label{Prop:SubGauExp1DReturn}
		Let $\mu \in \mathbb{R}$ and let $Y_1, Y_2, \ldots$ be independent random variables. Define $U_n := \sum_{i=1}^n Y_i$ for $n \in \mathbb{N}$ with $U_0 := 0$, and let $\xi:=\sargmin_{n \in \mathbb{N}_0} \mu U_n$.\\
		(a) Assume that $Y_1-\mu, Y_2-\mu, \ldots$ are independent sub-Gaussian random variables with variance parameter $1$.  Then
		\[ 
		\mathbb{P}(\xi \geq y) \leq 	\mathbb{P}\Bigl(  \inf_{n\in\mathbb{N}_0:n\geq y} \mu U_n \leq 0 \Bigr) \leq 3(\mu^{-2}\vee 1) e^{-y\mu^2/2}
		\]
		for $y \in [0, \infty)$.
		
		\noindent (b) Assume that $Y_1-\mu, Y_2-\mu, \ldots$ are independent sub-exponential random variables with variance parameter $1$ and rate parameter $A > 0$. Then
		\[ 
		\mathbb{P}(\xi \geq y) \leq 	\mathbb{P}\Bigl(  \inf_{n\in\mathbb{N}_0:n\geq y} \mu U_n \leq 0 \Bigr) \leq 3\Bigl( \frac{1}{\mu^2} \vee \frac{1}{\mu A} \vee 1\Bigr) e^{-y\mu(\mu\wedge A)/2}
		\]
		for $y \in [0, \infty)$.
	\end{prop}
	\begin{proof}
	Following the same argument at the beginning of the proof of Lemma~\ref{Lemma:1DLemma}, it suffices to only prove the latter inequality for both results for $\mu > 0$. \\
	(a) By a union bound and Proposition~\ref{Prop:SubGauExpTailBound}(a), we have
	\begin{align*}
	\mathbb{P}\Bigl(  \inf_{n\in\mathbb{N}_0:n\geq y} \mu U_n \leq 0 \Bigr) &\leq \sum_{n=\lceil y \rceil}^\infty  \mathbb{P}(U_n \leq 0) = \sum_{n=\lceil y \rceil}^\infty \mathbb{P}(U_n - n\mu \leq -n\mu) \leq \sum_{n=\lceil y \rceil}^\infty e^{-n\mu^2/2} \\
	&\leq \frac{e^{-y\mu^2/2}}{1-e^{-\mu^2/2}} \leq \frac{e^{-y\mu^2/2}}{ \frac{\mu^2}{4\log 2} \wedge \frac{1}{2} } \leq 3(\mu^{-2}\vee 1) e^{-y\mu^2/2}.
	\end{align*}
	(b) Again, by a union bound and Proposition~\ref{Prop:SubGauExpTailBound}(b), we have
	\begin{align*}
	\mathbb{P}\Bigl(  \inf_{n\in\mathbb{N}_0:n\geq y} \mu U_n \leq 0 \Bigr) &\leq \sum_{n=\lceil y \rceil}^\infty \mathbb{P}(U_n - n\mu \leq -n\mu) \leq \sum_{n=\lceil y \rceil}^\infty e^{-n\mu(\mu\wedge A)/2} \\
	&\leq 3\Bigl( \frac{1}{\mu^2} \vee \frac{1}{\mu A} \vee 1\Bigr) e^{-y\mu(\mu\wedge A)/2},
	\end{align*}
	where the last inequality follows from the last three inequalities in the proof for the (a) part, with $\mu(\mu\wedge A)$ taking the place of $\mu^2$.
	\end{proof}
	Our final preparatory result will be used to modify~\eqref{Eq:Prop5Eq1} in the proof of Proposition~\ref{prop:ocdsatisfyassumption}, which can be traced back to the proof of \citet[][Theorem~1]{CWS2021}. Let $b \neq 0$, $T^{\mathrm{diag}} >0$ and let $Z_1,Z_2,\ldots$ be independent centered random variables. We define the stopping time
	\begin{equation} \label{Eq:N_os}
	N_{\mathrm{os}} := \inf \Bigl\{n\in\mathbb{N}: b\sum_{t=1}^n(Z_t - b/2) \geq T^{\mathrm{diag}} \Bigr\},
	\end{equation}
	where `os' stands for \underline{o}ne-\underline{s}ided. When $Z_1,Z_2 \ldots \stackrel{\mathrm{iid}}{\sim} N(0,1)$, we can use a sequential probability ratio test argument to show that $\mathbb{P}(N_{\mathrm{os}} < \infty) \leq e^{-T^{\mathrm{diag}}}$.   The aim of Proposition~\ref{Prop:SubGauExpN_os} below is to establish similar bounds under the sub-Gaussian and the sub-exponential distributional assumptions respectively.
	\begin{prop} \label{Prop:SubGauExpN_os}
	    (a) Let $b \neq 0$, $T^{\mathrm{diag}} >0$ and let $Z_1, Z_2, \ldots$ be independent sub-Gaussian random variables each with variance parameter $1$. Then $N_{\mathrm{os}}$ defined in~\eqref{Eq:N_os} satisfies
	    \[
	    \mathbb{P}(N_{\mathrm{os}} < \infty) \leq e^{-T^{\mathrm{diag}}}. 
	    \]
	    (b) Let $T^{\mathrm{diag}} >0$ and $Z_1, Z_2, \ldots$ be independent sub-exponential random variables with variance parameter $1$ and rate parameter $A > 0$. Let $b \in [-A, A]\setminus\{0\}$. Then 
	    \[
	    \mathbb{P}(N_{\mathrm{os}} < \infty) \leq e^{-T^{\mathrm{diag}}}. 
	    \]
	\end{prop}
	\begin{proof}
	Without loss of generality, assume $b > 0$ in part~(a) and $0 < b < A$ in part~(b). The following argument then holds for both cases.  Denote $S_n := \sum_{t=1}^n Z_t$ with $S_0 := 0$, $V_n := \exp\{(bS_n- b^2n/2)\wedge T^{\mathrm{diag}}\}$ and $\mathcal{F}_n = \sigma(Z_1, \ldots, Z_n)$ with $\mathcal{F}_0$ defined to be trivial $\sigma$-algebra. Then for $n\in\mathbb{N}$, we have 
	\begin{align*}
	\mathbb{E}[V_n \mid \mathcal{F}_{n-1}] &= \mathbb{E}\bigl[V_n\mathbbm{1}_{\{V_{n-1} < e^{T^{\mathrm{diag}}}\}} \bigm| \mathcal{F}_{n-1}\bigr] + \mathbb{E}\bigl[V_n\mathbbm{1}_{\{V_{n-1} = e^{T^{\mathrm{diag}}}\}} \bigm| \mathcal{F}_{n-1}\bigr]\\
	&\leq \mathbb{E}\bigl[V_{n-1} e^{bZ_n - b^2/2}\mathbbm{1}_{\{V_{n-1} < e^{T^{\mathrm{diag}}}\}} \bigm| \mathcal{F}_{n-1}\bigr] + \mathbb{E}\bigl[e^{T^{\mathrm{diag}}}\mathbbm{1}_{\{V_{n-1} = e^{T^{\mathrm{diag}}}\}} \bigm| \mathcal{F}_{n-1}\bigr]\\
	&\leq V_{n-1}\mathbbm{1}_{\{V_{n-1} < e^{T^{\mathrm{diag}}}\}} + e^{T^{\mathrm{diag}}}\mathbbm{1}_{\{V_{n-1} = e^{T^{\mathrm{diag}}}\}} = V_{n-1}.
	\end{align*}
	Hence $(V_n)_{n \geq 0}$ is a supermartingale with respect to the filtration $(\mathcal{F}_n)_{n \geq 0}$. Since $|V_n| \leq e^{T^{\mathrm{diag}}}$ for all $n$, we have by the Optional Stopping Theorem that
	\begin{align*}
	    1 = \mathbb{E}V_0 \geq \mathbb{E}V_{N_{\mathrm{os}}} = \mathbb{E}\bigl[ V_{N_{\mathrm{os}}} \mathbbm{1}_{\{N_{\mathrm{os}} < \infty\}} + V_{N_{\mathrm{os}}} \mathbbm{1}_{\{N_{\mathrm{os}} = \infty\}}\bigr] \geq e^{T^{\mathrm{diag}}}\mathbb{P}(N_{\mathrm{os}} < \infty),
	\end{align*}
	as required.
	\end{proof}
	The modifications to the proofs of Theorems~\ref{Thm:Coverage},~\ref{Thm:Length} and~\ref{Thm:Support}, as well as Proposition~\ref{prop:ocdsatisfyassumption} that are needed in the sub-Gaussian and sub-exponential settings are as follows.  In~\eqref{Eq:CoveragePfEq1}, \eqref{Eq:Orange1}, \eqref{Eq:Orange2}, \eqref{Eq:Orange3}, \eqref{Eq:gensig_randomN_omega3}, \eqref{Eq:Orange4}, \eqref{Eq:Orange5}, \eqref{Eq:ocdassumption_term1+}, \eqref{Eq:Orange6}, \eqref{Eq:Normal_W1} and~\eqref{Eq:Normal_W2}, we apply Proposition~\ref{Prop:SubGauExpTailBound} in place of the Gaussian tail bounds; in \eqref{Eq:CoveragePfEq2},
\eqref{Eq:gensig_randomN_omega1} and
\eqref{Eq:SmallProbShortTail}, we apply Proposition~\ref{Prop:SubGauExp1DReturn} in place of Lemma~\ref{Lemma:1DLemma}; and finally, we apply Proposition~\ref{Prop:SubGauExpN_os} to yield the bound  corresponding to~\eqref{Eq:Prop5Eq1}.

	\subsection{Additional simulation results}
	\label{Sec:AddSims}
	
	Table~\ref{Tab:Spatial} provides additional simulation results for the \texttt{ocd\_CI} procedure under spatial dependence.  The data generating mechanisms and conclusions from these results are given in Section~\ref{Sec:CoverageLength}.
		
\begin{table}[htbp]
\footnotesize
\begin{center}
\caption{\label{Tab:Spatial} Spatial dependence.  Estimated coverage and average length of the \texttt{ocd\_CI} confidence interval and average detection delay over 2000 repetitions, with standard errors in brackets, under a Toeplitz cross-sectional covariance matrix $\Sigma$ with entries $\Sigma_{jk} = \rho^{|j-k|}$ for $j,k \in [p]$.  Other parameters: $p=100$, $\beta=\vartheta$, $\gamma=30000$, $z=1000$, $\alpha = 0.05$, $a = \tilde{a} = \sqrt{2 \log p}$, $c=0.5$, $d_1 = c\sqrt{\log(p/\alpha)}$, $d_2 = 4d_1^2$.}
\begin{tabular}{cccccc}
\hline\hline
$\rho$ & $s$ & $\vartheta$ & Detection Delay & Coverage (\%) & CI Length\\
\hline
$0.5$ & $2$ & $2$ & $13.9_{(0.1)}$ & $98.5_{(0.3)}$ & $35.5_{(1.0)}$\\
$0.5$ & $2$ & $1$ & $49.1_{(0.3)}$ & $99.0_{(0.2)}$ & $125.1_{(1.6)}$\\
$0.5$ & $2$ & $0.5$ & $172.5_{(1.0)}$ & $99.5_{(0.2)}$ & $447.0_{(2.8)}$\\
$0.5$ & $10$ & $2$ & $21.9_{(0.1)}$ & $98.7_{(0.3)}$ & $42.0_{(0.9)}$\\
$0.5$ & $10$ & $1$ & $76.1_{(0.5)}$ & $98.8_{(0.2)}$ & $154.2_{(1.5)}$\\
$0.5$ & $10$ & $0.5$ & $266.7_{(1.8)}$ & $99.0_{(0.2)}$ & $566.9_{(3.9)}$\\
$0.5$ & $100$ & $2$ & $52.1_{(0.3)}$ & $98.3_{(0.3)}$ & $106.8_{(0.9)}$\\
$0.5$ & $100$ & $1$ & $187.7_{(1.3)}$ & $98.4_{(0.3)}$ & $399.5_{(3.3)}$\\
$0.5$ & $100$ & $0.5$ & $655.3_{(5.0)}$ & $98.5_{(0.3)}$ & $1366.2_{(10.1)}$\vspace{1.5mm}\\
$0.75$ & $2$ & $2$ & $13.9_{(0.1)}$ & $96.9_{(0.4)}$ & $51.1_{(2.6)}$\\
$0.75$ & $2$ & $1$ & $47.9_{(0.3)}$ & $96.8_{(0.4)}$ & $146.0_{(3.3)}$\\
$0.75$ & $2$ & $0.5$ & $171.5_{(1.1)}$ & $97.7_{(0.3)}$ & $463.8_{(4.2)}$\\
$0.75$ & $10$ & $2$ & $21.8_{(0.2)}$ & $96.4_{(0.4)}$ & $48.6_{(1.7)}$\\
$0.75$ & $10$ & $1$ & $75.3_{(0.5)}$ & $96.7_{(0.4)}$ & $165.0_{(2.7)}$\\
$0.75$ & $10$ & $0.5$ & $266.3_{(1.9)}$ & $96.0_{(0.4)}$ & $558.9_{(4.5)}$\\
$0.75$ & $100$ & $2$ & $50.9_{(0.3)}$ & $96.8_{(0.4)}$ & $106.8_{(1.2)}$\\
$0.75$ & $100$ & $1$ & $184.8_{(1.4)}$ & $95.6_{(0.5)}$ & $401.8_{(3.8)}$\\
$0.75$ & $100$ & $0.5$ & $647.3_{(5.4)}$ & $94.6_{(0.5)}$ & $1312.3_{(11.2)}$\\
\hline
\end{tabular}
\end{center}
\end{table}

\textbf{Acknowledgements:} The research of TW was supported by Engineering and Physical Sciences Research Council (EPSRC) grant EP/T02772X/1 and that of RJS was supported by EPSRC grants EP/P031447/1 and EP/N031938/1, as well as European Research Council Advanced Grant 101019498.
	
	\bibliographystyle{custom2author}
    \bibliography{Change_CI}
\end{document}